\documentclass[11pt]{article}

\pdfoutput=1

\usepackage{geometry}
\usepackage{
	amsmath,
	amsthm,
	thmtools,
	thm-restate,
	mathpartir
}

\usepackage{amsfonts,amssymb,mathtools}

\usepackage[labelfont=bf, font=small]{caption}

\usepackage{authblk}

\usepackage{url}
\usepackage{hyperref}
\hypersetup{
	colorlinks=true,
	citecolor=magenta,
	linkcolor=magenta,
	urlcolor=magenta
}

\usepackage[capitalize]{cleveref}

\usepackage{tikz}
\usetikzlibrary{automata,arrows,positioning}
\tikzset{
	->, 
	>={stealth},
	every state/.style = {thick},
	every edge/.append style = {every node/.style={font=\scriptsize}},
	node distance = {2cm},
	initial text = {\( \)}
}

\usepackage{tikz-cd}

\usepackage{prftree}

\usepackage{enumitem}

\usepackage[
	style=ieee
]{biblatex}
\title{A (Co)Algebraic Framework for Ordered Processes\thanks{This work was partially supported by ERC grant Autoprobe (grant agreement 101002697).}}

\author{Todd Schmid\footnote{\url{todd.schmid.19@ucl.ac.uk}}}

\affil{UCL,	London, UK}

\date{}

\begin{document}

\theoremstyle{plain}
\newtheorem{theorem}{Theorem}[section]
\newtheorem{lemma}[theorem]{Lemma}
\newtheorem{corollary}[theorem]{Corollary}
\newtheorem{proposition}[theorem]{Proposition}

\theoremstyle{definition}
\newtheorem{remark}[theorem]{Remark}
\newtheorem{definition}[theorem]{Definition}
\newtheorem{example}[theorem]{Example}


\theoremstyle{definition}
\newtheorem{assumption}{Assumption}
\newtheorem*{assumption*}{Assumption}
\newtheorem{question}{Question}

\tikzset{every state/.style={minimum size=0pt}}
\tikzset{every node/.style={scale=1}}

\makeatletter
\newcommand*{\da@rightarrow}{\mathchar"0\hexnumber@\symAMSa 4B }
\newcommand*{\da@leftarrow}{\mathchar"0\hexnumber@\symAMSa 4C }
\newcommand*{\xdashrightarrow}[2][]{%
  \mathrel{%
    \mathpalette{\da@xarrow{#1}{#2}{}\da@rightarrow{\,}{}}{}%
  }%
}
\newcommand*{\da@xarrow}[7]{%
  \sbox0{$\ifx#7\scriptstyle\scriptscriptstyle\else\scriptstyle\fi#5#1#6\m@th$}%
  \sbox2{$\ifx#7\scriptstyle\scriptscriptstyle\else\scriptstyle\fi#5#2#6\m@th$}%
  \sbox4{$#7\dabar@\m@th$}%
  \dimen@=\wd0 %
  \ifdim\wd2 >\dimen@
    \dimen@=\wd2 %
  \fi
  \count@=2 %
  \def\da@bars{\dabar@\dabar@}%
  \@whiledim\count@\wd4<\dimen@\do{%
    \advance\count@\@ne
    \expandafter\def\expandafter\da@bars\expandafter{%
      \da@bars
      \dabar@ 
    }%
  }%
  \mathrel{#3}%
  \mathrel{%
    \mathop{\da@bars}\limits
    \ifx\\#1\\%
    \else
      _{\copy0}%
    \fi
    \ifx\\#2\\%
    \else
      ^{\copy2}%
    \fi
  }%
  \mathrel{#4}%
}
\makeatother

\newcommand 	\sle 		{\sqsubseteq} 		
\newcommand 	\sge 		{\sqsupseteq}		

\newcommand   \theory[1]{\mathsf{#1}}

\newcommand 	\At  		  {A} 		
\newcommand 	\Eq 		  {\theory{E}}		
\newcommand 	\Ie 		  {\theory{IN}} 		
\newcommand 	\Exp 		  {\mathsf{Exp}} 		
\newcommand 	\Expm 		{{\Exp}/{\equiv}}	
\newcommand 	\Expi 		{{\Exp/{\sle}}}		
\newcommand 	\SExp 		{\mathsf{SExp}} 	
\newcommand 	\PExp 		{\mathsf{PExp}} 	
\newcommand 	\acro[1]	{\(\mathsf{#1}\)} 				

\newcommand 	\Sets		  {\mathbf{Sets}}				
\newcommand 	\Met 		  {\mathbf{Met}} 				
\newcommand 	\Pos		  {\mathbf{Pos}} 				
\newcommand  	\dCPO 		{\mathbf{dCPO}}				
\newcommand 	\dcpo  		{{\textsf{dcpo}}}
\newcommand 	\Cat 		  {\mathbf{C}}
\newcommand 	\Alg 		  {\operatorname{Alg}}
\newcommand 	\Coalg 		{\operatorname{Coalg}}  	

\newcommand 	\Id 		{\operatorname{Id}}					
\renewcommand 	\P		{\mathcal{P}_{\omega}}		
\newcommand   \A      {\mathcal A_\omega}    
\newcommand 	\M			{\mathcal{M}_{\omega}}	
\newcommand 	\N 			{\mathbb N}				
\newcommand 	\D 			{\mathcal{D}_\omega}	
\newcommand 	\C			{\mathcal{C}}			
\newcommand 	\PA 		{\mathsf{PA}}

\newcommand 	\E 			{{\mathsf{Prc}}}
\newcommand 	\X 			{{\mathbb X}}
\newcommand 	\Y 			{{\mathbb Y}}

\newcommand 	\id 			  {\operatorname{id}}			
\newcommand 	\supp 			{\operatorname{supp}}		
\newcommand 	\conv 			{\operatorname{conv}}	
\newcommand 	\fv 			  {\operatorname{fv}}	
\newcommand 	\bv 			  {\operatorname{bv}}		
\newcommand 	\downset 		{\operatorname{\downarrow}}
\newcommand 	\Pfp 			  {\operatorname{Pfp}}
\newcommand 	\lfp 			  {\operatorname{lfp}}
\newcommand 	\fp 			  {\operatorname{fp}}
\newcommand 	\disc  			{\operatorname{Dis}}
\newcommand 	\forg[1] 		{\left|#1\right|}

\newcommand 	\sub 			{\mathbin{/}}
\newcommand 	\gdsub			{\mathbin{/\!/}}
\newcommand 	\unsub 			{\raisebox{-0.3em}{
	\begin{tikzpicture}
		\draw[thick, dotted, >={}] (0,-0.2) -- (1ex, 1.2ex);
	\end{tikzpicture}
}}

\newcommand 	\upper 		{\operatorname{\uparrow}}
\newcommand 	\lhdeq  	{\unlhd}

\newcommand 	\qm 		{\mathbin{?}}

\newcommand 	\tr[1] 		{\mathrel{\raisebox{-0.1em}{{\footnotesize\(\xrightarrow{#1}\)}}}}
\newcommand 	\out[1] 	{\mathrel{\raisebox{-0.1em}{{\footnotesize\(\xRightarrow{#1}\)}}}}
\newcommand 	\trd[1] 	{\mathrel{\raisebox{-0.1em}{{\footnotesize\(\xdashrightarrow{#1}{}\)}}}}

\newcommand 	\Act 		  {\mathsf{Act}} 		
\newcommand   \act      {\text{act}}
\newcommand 	\Var 		  {\mathsf{Var}}		
\newcommand   \var      {\text{var}}
\newcommand 	\code[1] 	{\mathsf{#1}}

\newcommand 	\sem[1] 	{\lceil \!\! \lfloor #1 \rceil \!\! \rfloor}
\newcommand 	\asem[1] 	{\lceil \!\! \lceil #1 \rceil \!\! \rceil}
\newcommand 	\ev 		{\operatorname{\mathsf{ev}}}
\newcommand 	\beh   		{{!}}
\newcommand 	\gd			{\mathsf{gd}} 				
\newcommand 	\skiptt  	{\code{1}} 			
\newcommand 	\failtt 	{\code{0}}
\newcommand 	\unit		{\underline{\code{u}}}
\newcommand 	\eff  		{\code{f}}

\newcommand   \bisim    {\mathrel{\underline{\leftrightarrow}}}

\newcommand 	\defn[1] 	{\emph{#1}}
\newcommand   \uaequiv {\mathrel{\dot\equiv}}

\maketitle

\begin{abstract}
A recently published paper~\cite{schmidrozowskisilvarot2022processes}
offers a (co)algebraic framework for studying processes with monadic branching structures and recursion operators. 
The framework captures Milner's algebra of regular behaviours~\cite{milner1984complete},
but fails to give an honest account of a closely related calculus of probabilistic processes~\cite{stark1999complete}
We capture Stark and Smolka's calculus by giving an alternative framework, aimed at studying a family of ordered process calculi with inequationally specified branching structures and recursion operators.
We observe that the recent probabilistic extension of guarded Kleene algebra with tests~\cite{rozowskikappekozenschmidsilva2022probgkat}
is a fragment of one of our calculi, along with other examples. 
We also compare the intrinsic order in our process calculi with the notion of similarity in coalgebra. 
\end{abstract}



\section{Introduction}
Process calculi are commonly used models of nondeterminism and concurrency~\cite{baeten2005history} and are known to be powerful tools for specifying and reasoning about stateful transition systems in general~\cite{silva2010kleene,myers2009coalgebraic}.
Frameworks that capture these process types can therefore (a) help us understand the connections between different calculi and (b) guide the development of new process calculi by exhibiting generalisable design principles for these calculi. 

In ~\cite{schmidrozowskisilvarot2022processes} it is shown that a large family of process calculi can be studied with the \emph{processes parametrised framework}, which provides a generic toolset for studying transition systems with algebraically specified branching strutures.
The processes parametrised framework is motivated by two important examples of process calculi: the \emph{algebra of regular behaviours} (or \acro{ARB}) introduced by Milner~\cite{milner1984complete} for specifying labelled transition systems and the \emph{algebra of probabilistic actions} (or \acro{APA})~\cite{schmidrozowskisilvarot2022processes}, based on the probabilistic calculus of Stark and Smolka~\cite{stark1999complete} for specifying probabilistic labelled transition systems. 
However, while the processes parametrised framework captures the algebra of regular behaviours, it fails to capture Stark and Smolka's semantics of recursion. 

In the processes paramterised framework, unguarded recursion calls are interpreted as deadlock.
For example, consider the following recursive specifications of a process called \(\code x\).
\begin{align}
	\code x &= \code{if (flip(heads))} ~\{ ~
			\code{run}~\code x; 
			\}~\code{ else } ~\{~ 
				\code{print }~ ``\texttt{tails}"; ~ \code{run}~\code x;
			\} \\
	\code x &= \code{if (flip(heads))} ~\{ ~
			\code{halt}; 
			\}~\code{ else } ~\{~ 
				\code{print }~ ``\texttt{tails}"; ~ \code{run}~\code x;
			\} \\
	\code x &= \code{print }~ ``\texttt{tails}"; \code{run}~\code x;
\end{align}
Specification (1) is essentially a loop with an unguarded recursion call: it flips a coin to determine whether to recurse or to print ``tails'' and start over. 
Specification (2) flips a coin to determine whether to enter deadlock or print the word ``tails'' and start over, and (3) simply repeatedly prints ``tails''.
All three specifications exist in \acro{APA} as well as in the processes parametrised framework.
Intuitively, specification (1) is the same as specification (3), as any repeated coin toss eventually sees a tails-up coin with probability \(1\). 
Only Stark and Smolka's calculus identifies specifications (1) and (3).
The semantics provided by the processes parametrised framework, on the other hand, replaces the unguarded recursive call in (1) with deadlock, producing specification (2). 

The difference between Stark and Smolka's calculus and the processes parametrised framework is that the former interprets unguarded recursion as a least fixed point.
In this paper, we propose the \emph{ordered processes parametrised framework}, which provides a generic toolset for specifying and reasoning about processes with inequationally specified branching that admits least fixed points for unguarded recursive specifications. 
It consists of a specification language for finite processes, a small-step semantics for executing specifications, a denotational semantics, and a sound and complete axiomatisation.

Notably, our framework also captures the \emph{algebra of control flows} introduced in~\cite{schmidrozowskisilvarot2022processes} for specifying propositional imperative programs~\cite{kozentseng2008,gkat,schmidkappekozensilva2021gkat}. 
Milner's algebra of regular behaviours does not find an exact instantiation, however, so the ordered processes parametrised framework is a variation on, not a generalisation of, the processes parametrised framework.  


One of the most fascinating features of Milner's \acro{ARB}~\cite{milner1984complete} is its fragment resembling the algebra of regular expressions.
Axiomatising this fragment has been the subject of intense research~\cite{baetenCG2006starheight,baetenCG2007characterisation,fokkink1997perpetual,fokkinkzantema1997termination,grabmayerfokkink2020complete,grabmayer2021coinductive}, and only recently has a proof been announced that a Salomaa-style axiomatisation is complete.
The processes parametrised framework~\cite{schmidrozowskisilvarot2022processes} proposes that a version of Milner's regular expressions can be found in many other process calculi.
Kleene algebra variants such as \acro{GKAT}~\cite{gkat,schmidkappekozensilva2021gkat} can also be obtained as star fragments. 
Process calculi covered by the ordered processes parametrised framework also have star fragments, and the recently introduced probabilistic programming language \acro{ProbGKAT}~\cite{rozowskikappekozenschmidsilva2022probgkat} can be obtained as a star fragment of one of our ordered process calculi. 

The processes parametrised framework proposes a uniform axiomatisation of star fragments, conjectured in~\cite{schmidrozowskisilvarot2022processes} to be complete.
While star fragments appear in our process calculi syntactically, it is unclear how to give a similarly uniform axiomatisation. 
The issue is that resolving unguarded recursive calls in a star expression (called \emph{loop tightening} in \acro{GKAT}~\cite{gkat}, and \(A_{11}\) in Salomaa's first axiomatisation of Kleene algebra~\cite{salomaa1966two}) can result in an expression that lies outside of the star fragment. 
We propose another fragment, the \emph{polystar} fragment, to overcome this technical difficulty and suggest a uniform axiomatisation. 

Our development of the ordered processes parametrised framework uses the theoretical toolboxes of universal algebra~\cite{cohn1981universal} and coalgebra~\cite{rutten2000universal,adamek2005introduction}.
We specifically rely on the theory of ordered algebras~\cite{bloomwright1983pvarieties,kurzvelebil2017quasivarieties} for the interpretation of least fixed points, and on the theory of ordered coalgebras~\cite{levy2011similarity,kapulkinkurzvelebil2012coalgposet} for transition systems with partially ordered states. 
Universal coalgebra proposes well-motivated and intuitive notions of behavioural equivalence for processes, and by adopting ordered coalgebraic models, it furthermore proposes a partial order on system behaviours that we incorporate into our semantics.

The partial order on system behaviours is related but not equivalent to \emph{similarity}~\cite{hughesjacobs2004simulations}.
We shed new light on an old result in automata theory, that two-way similarity and bisimilarity do not coincide for nondeterministic automata~\cite{bloomIM1995traced}, by viewing it as a property of the algebraic theory underlying the branching structure of nondeterministic processes. 

In short, we propose a framework for studying process calculi that specify ordered processes with inequationally characterised branching structures admitting recursion, instantiate the framework to several examples, and discuss its connection to existing tools for coalgebraic models of process algebra.
Our contributions are organised as follows: 
\begin{itemize}
	\item We introduce a family of ordered process types parametrised by inequational theories in \cref{sec:ordered processes} and \cref{sec:inequational theories} and order their behaviours. 
	\item We equip each process type whose inequational theory admits a form of recursion with a specification language in \cref{sec:a family of ordered process calculi} and connect the resulting algebra with an algebra of regular behaviours in \cref{sec:semantics}, where an adequacy result can be found.
	\item We give a complete axiomatisation of the behaviour order in \cref{sec:axiomatisation}.
	\item We introduce the naturally occurring \emph{polystar fragments} of our calculi in \cref{sec:star fragments} and propose sound axioms for their behavioural order.
	\item Finally, we tie the notion of similarity from coalgebra to the behaviour order in \cref{sec:simulations} and study two-way similarity using the inequational theories of \cref{sec:inequational theories}.
\end{itemize}
A discussion of related literature and further research directions can be found in \cref{sec:discussion}.

\section{Ordered Processes}\label{sec:ordered processes} 
We start with the introduction of a family of monotone labelled transition system types and a formal description of their behaviours.
Roughly, a monotone labelled transition system is a system whose states carry a partial order, with respect to which the transition structure is monotone.
We formalise these in the language of universal coalgebra~\cite{rutten2000universal,adamek2005introduction}, which captures transition types as endofunctors on a category.

We work in the category \(\Pos\) of posets and monotone maps.
A poset is a pair \((X, \le)\) consiting of an \emph{underlying set} \(X\) and a partial order \({\le} \subseteq X^2\).
We write \(h : (X, \le) \to (Y, \le)\) to denote that \(h : X \to Y\) is monotone. 
For a given endofunctor \(B\) on \(\Pos\), a (\emph{monotone \(B\)-})\emph{coalgebra} is a triple \((X, \le, \vartheta)\) consisting of a set \(X\) of \emph{states}, a partial order \(\le\) on \(X\) called the \emph{state order}, and a \emph{structure map} \(\vartheta : (X, \le) \to B(X, \le)\). 
A \emph{coalgebra homomorphism} \(h : (X, \le, \vartheta) \to (Y, \le, \delta)\) is a map \(h : (X, \le) \to (Y, \le)\) such that \(\delta \circ h = B(h) \circ \vartheta\). 

An \emph{embedding} is a monotone map that reflects the partial order, and a \emph{coalgebra embedding} is a colagebra homomorphism that is also an embedding.
A subset \(U \subseteq X\) is called a \emph{subcoalgebra} of \((X, \le, \vartheta)\) if the inclusion map is a coalgebra embedding \((U, \le, \vartheta_U) \hookrightarrow (X, \le, \vartheta)\) for some \(\vartheta_U : (U, \le) \to B(U, \le)\).
Subcoalgebras play an important role in \cref{sec:axiomatisation}.  

\begin{example}\label{eg:brzozowski}
	Letting \(\Id\) denote the identity, we define \(\Id^A : \Pos \to \Pos\) via the product order: given \((X, \le)\), set \({(X, \le)^A} = (X^A, \le^A)\), where \(X^A = \{f \mid f : A \to X\}\) and
	\(
		f \le^A g
	\)
	if and only if
	\( 
		{(\forall a \in A)}~f(a) \le g(a)
	\).
	Given \(h : {(X, \le)} \to (Y, \le)\), we define \(h^A : {(X^A, \le^A)} \to (Y^A, \le^A)\) by \(h^A (f) =  h \circ f\).
	Regular languages with their Brzozowski derivatives~\cite{brzozowski1964derivatives} constitute a monotone \(\Id^A\)-coalgebra \(D : (\text{Reg}, \subseteq) \to (\text{Reg}, \subseteq)^A\), where \(D(L)(a) = a^{-1}L\).
\end{example}

Let \(M : \Pos \to \Pos\) be an endofunctor called the \emph{branching type}, and fix an infinite discrete poset \((\Var, =)\) of \emph{return variables} and a poset \((\Act, \le_{\act})\) of \emph{actions} (or \emph{commands}).
We take return variables to be discretely ordered to ensure that substitution of return variables is monotone.

\begin{definition}	
	The \emph{\(M\)-branching system type} is \(B_M = M((\Var, =) + (\Act, \le_{\act})\times \Id)\).
\end{definition}
We refer to \(B_M\)-coalgebras as \emph{\(M\)-branching monotone labelled transition systems}, or \emph{monotone systems} for short. 
Intuitively, monotone systems are a variation on the labelled transition systems studied in~\cite{milner1984complete}, with extra decorations to denote branching structure. 

\begin{example}\label{eg:GKAT-automata}
	\emph{Monotone \acro{GKAT}-automata (with return variables)} are \(B_M\)-coalgebras for \(M = \Id_\bot^A := (\Id_\bot)^A\).
	Here, \((X, \le)_\bot = (X \sqcup \{\bot\}, \le_\bot)\) is defined so that \(x \le_\bot y\) if and only if either \(x \le y\) in \((X, \le)\) or \(x = \bot\), and \(A\) is a set of \emph{atomic tests}. 
	\acro{GKAT}-automata are the coalgebraic models for the while-program fragment of Kleene algebra with tests~\cite{kozensmith1996kat}, proposed first in~\cite{kozentseng2008} and then revisited in~\cite{gkat}.
	Below is a monotone \acro{GKAT}-automaton on a three-element poset, with blue edges denoting its Hasse diagram.
	
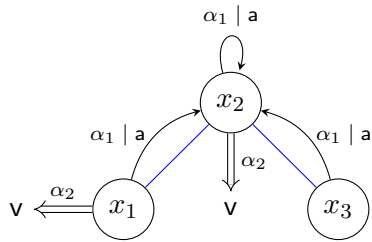
\[\begin{tikzpicture}[baseline=(current bounding box.center)]
		\node[state] (0) {\(x_1\)};
		\node[state, above right of= 0] (1) {\(x_2\)};
		\node[state, below right of= 1] (2) {\(x_3\)};
		\node[left=2em of 0] (3) {\(\code v\)};
		\node[below=2em of 1] (4) {\(\code v\)};

		\draw (0) edge[-, blue] (1);
		\draw (2) edge[-, blue] (1);
		\draw (0) edge[->, bend left] node[left] {\(\alpha_1 \mid \code a\)} (1);
		\draw (2) edge[->, bend right] node[right] {\(\alpha_1 \mid \code a\)} (1);
		\draw[loop] (1) edge[->, loop above] node[above] {\(\alpha_1 \mid \code a\)} (1);
		
		\draw (0) edge[-implies, double, double distance=2pt] node[above] {\(\alpha_2\)} (3);
		\draw (1) edge[-implies, double, double distance=2pt] node[right] {\(\alpha_2\)} (4);
	\end{tikzpicture}
	\qquad\qquad
	\begin{aligned}
		X &= \{x_1, x_2, x_3\} \\
		x_i \tr{\alpha \mid \code a} x_j &\iff \vartheta(x_i)(\alpha) = (a, x_j) \\
		x_i \out{\alpha} \code v &\iff \vartheta(x_i)(\alpha) = \code v
	\end{aligned}\]
\end{example}

\begin{example}\label{eg:probabilistic processes}
	\emph{Monotone probabilistic labelled transition systems} are \(B_M\)-coalgebras with \(M = \D^\bot\), defined by \(\D^\bot(X, \le) = (\D^\bot X, \le^{hh})\) where 
	\[
		\D^\bot X = \left\{\theta : X \to [0,1]~\middle|~\begin{array}{c}|\{x \mid \theta(x) > 0\}| < \infty \\\sum_{x \in X}\theta(x) \le 1\end{array}\right\}
		\qquad 
		\D^\bot(h)(\theta)(y) = \sum_{h(x)= y} \theta(x)
	\]
	for any \(h : (X, \le) \to (Y, \le)\). 
	We call \(\le^{hh}\) the \emph{heavier-higher} order\footnote{It is one of the well-known \emph{stochastic} orders from statistics~\cite{shaked1994stochastic, mullerstoyan2002}.}, and define \(\theta_1 \le^{hh} \theta_2\) if and only if for any upwards-closed subset \(U \subseteq X\) (a set is upwards-closed if \(x \in U\) and \(x \le y\) implies \(y \in U\)), \(\sum_{x \in U}\theta_1(x) \le \sum_{x \in U}\theta_2(x)\). 
	Probabilistic processes appear in~\cite{stark1999complete,schmidrozowskisilvarot2022processes}. 
\end{example}

\begin{example}\label{eg:labelled transition systems}
	\emph{Monotone LTSs} are \(B_M\)-coalgebras for the \emph{downset functor} \(M = \A\), defined
	\[
		\A(X, \le) = (\{\downset \{x_1, \dots, x_n\} \mid x_1, \dots, x_n \in X\}, \subseteq)
		\qquad 
		\A(h)(U) = \downset h[U]
	\]
	where \(\downset U = \{z \mid (\exists x \in U)~z \le x\}\) for any \(U \subseteq X\), and \(h : (X, y) \to (Y, \le)\).  
\end{example}

\begin{remark}
	The transition type for ordered LTSs might look different than expected:
	in~\cite{milner1984complete}, for example, LTSs have \(M = \P\) as their branching structure~\cite{schmidrozowskisilvarot2022processes}. 
	Unlike \cref{eg:GKAT-automata} and \cref{eg:probabilistic processes}, \(\A\) is not a lifting\footnote{An endofunctor \(M\) on \(\Pos\) is a \emph{lifting} of \(N : \Sets \to \Sets\) if \(M(X, \le) = (N, \le)\) for any poset \((X, \le)\).} of the branching structure of its ordinary counterpart.
	Despite this, \({\A (X, =)} = {(\P X, \subseteq)}\), so discrete ordered LTSs and ordinary LTSs coincide when \(\le_{\act}\) is a discrete partial order.
\end{remark}

We use the phrases \emph{system}, \emph{monotone system}, \emph{labelled transition system}, and \emph{\(B_M\)-coalgebra} interchangeably.
The word \emph{process} informally refers to a system that has an initial state. 

\subsection*{The Behaviour Order}\label{sub:the behaviour order} 
Following~\cite{rutten2000universal}, the behaviours of coalgebras are compared with coalgebra homomorphisms. 
A state \(x_1\) of \((X_1, \le, \vartheta_1)\) is \emph{behaviourally below} another state \(x_2\) of \((X_2, \le, \vartheta_2)\), written \(x_1 \le_\beh x_2\), if there is a third coalgebra \((Y, \le, \delta)\) and coalgebra homomorphisms \(h_i : {(X_i, \le, \beta_i)} \to {(Y, \le, \delta)}\), \(i=1,2\), such that \(h_1(x_1) \le h_2(x_2)\).
Similarly, the state \(x_1\) is \emph{behaviourally equivalent} to \(x_2\), denoted \(x_1 =_\beh x_2\), if there are \(h_1,h_2\) such that \(h_1(x_1) = h_2(x_2)\).

\begin{restatable}{lemma}{behequivisanequiv}
	Let \(x_i\) be a state of \((X_i, \le, \vartheta_i)\) for \(i=1,2\). 
	Then \(x_1 =_\beh x_2\) if and only if both \(x_1 \le_\beh x_2\) and \(x_1 \le_\beh x_2\). 
\end{restatable}

Every branching structure we consider in this paper admits a \emph{final} monotone system\footnote{We only deal with the case in which \(B_M\) is finitary, which guarantees the existence of final coalgebras~\cite{barr1993terminal,adamekkoubek1995greatest}.} \((Z, \le_\beh, \zeta)\), a unique (up to isomorphism) \(B_M\)-coalgebra such that every \((X, \le, \vartheta)\) admits a unique coalgebra homomorphism \(\beh_\vartheta : (X, \le, \vartheta) \to (Z, \le_\beh, \zeta)\).
States of \((Z, \le_\beh, \zeta)\) are called \emph{behaviours} and the map \(\beh_\vartheta\) is often called \emph{unrolling}, because of the way it is constructed for automata~\cite{worrell1999terminal,rutten1998coinduction}. 
Unrolling allows us to compare behavioural inequivalence directly, ie. \(x_1 \le_\beh x_2\) if and only if \(\beh_{\vartheta_1}(x_1) \le_\beh \beh_{\vartheta_2}(x_2)\) in \((X_1, \le, \vartheta_1)\) \((X_2, \le, \vartheta_2)\)~\cite{kapulkinkurzvelebil2012coalgposet}.

\section{Inequational Theories}\label{sec:inequational theories} 
In this section, we give an overview of the inequational theories we use to characterise monotone branching structures.

System types covered by the processes parametrised framework are parametrised by equational algebraic theories.
This is motivated by the observation that many branching structures for \emph{ordinary} systems, coalgebras for functors of the form \(N(\Var + \Act \times \Id)\) on \(\Sets\) (see \cref{sec:simulations}), are instances of free algebra constructions for algebraic theories axiomatising term equivalences. 
For example, the branching structure of LTSs is \(\P\), which constructs the free semilattice with bottom in \(\Sets\). 
Similarly, branching structures for monotone system types are free algebra constructons for algebraic theories axiomatising term inequivalences.

A (\emph{finitary}) \emph{inequational theory} is a pair \((S, \Ie)\) consisting of a \emph{signature}, a polynomial functor of the form \(S = \coprod_{n \in \N} (I_n, \le_n) \times \Id^n\) on \(\Pos\), and a set \(\Ie\) of inequations for terms in the signature \(S\)~\cite{bloomwright1983pvarieties,kurzvelebil2017quasivarieties}.
An \emph{\(S\)-algebra} is a triple \((X, \le, \alpha)\) consisting of a poset \((X, \le)\) and a monotone \emph{operation map} \(\alpha : S(X, \le) \to (X, \le)\).
A map \(h : (X, \le) \to (Y, \le)\) between \(S\)-algebras \((X, \le, \alpha)\) and \((Y, \le, \gamma)\) is an \emph{\(S\)-algebra homomorphism} if \(h \circ \alpha = \gamma\circ S(h)\). 

\begin{remark}\label{rem:peculiar operation order}
	One peculiarity of this depiction of inequational theories allows operations to be partially ordered and forces operation maps to be monotone with respect to this order. 
	This is a generalisation due to Kurz and Velebil~\cite{kurzvelebil2017quasivarieties} of Bloom and Wright's inequational theories~\cite{bloomwright1983pvarieties}, where every \(I_n\) is discrete. 
\end{remark}

Formally, \(\Ie\) is a set of pairs of \emph{\(S\)-terms}, elements of a poset \(S^*(V, \le_V) = (S^*V, \le_S)\) whose underlying set is freely generated by \(V\) and the operations of \(S\).
The partial order \(\le_S\) is inferred from the \emph{\(S\)-term ordering rules} below:
\[ 
	\inferrule{(\forall i)~p_i \le_S q_i}{\sigma(\vec p) \le_S \sigma(\vec q)}
	\qquad 
	\inferrule{\sigma_1 \le_n \sigma_2}{\sigma_1(\vec p) \le_S \sigma_2(\vec p)} \qquad\qquad (\vec p = (p_1, \dots, p_n),~ \vec q = (q_1, \dots, q_n))
\]
Given a tuple \(\vec x = (x_1, \dots, x_n) \in X^n\), we write \(p = p(\vec x)\) to denote that the free variables of \(p\) are contained in \(\{x_1, \dots, x_n\}\), and for any \(\vec q = (q_1, \dots, q_n)\) we write \(p(\vec q)\) to denote the \(S\)-term obtained by replacing every \(x_i\) in \(p\) with \(q_i\) for each \(i \le n\). 
An \(S\)-algebra \((X, \le, \alpha)\) \emph{satisfies} \(\Ie\) (is an \emph{\((S, \Ie)\)-algebra}) if \(p(\vec x) \le q(\vec x)\) for any \((p(\vec v), q(\vec v)) \in \Ie\) and any \(\vec x \in X^n\).

Given \(p,q \in S^*X\), write \(p \sqsubseteq_\Ie q\) if \(p \le q\) can be proven from the \(S\)-term ordering rules and \(\Ie\), and write \(p \equiv_\Ie q\) if \(p \sqsubseteq_\Ie q \sqsubseteq_\Ie p\). 

\begin{definition}\label{def:free SIalg}
	The \emph{free} \((S, \Ie)\)-algebra \emph{on \((X, \le)\)} is the poset \(M(X, \le)\) of \(\equiv_\Ie\)-equivalence classes \([p]_\Ie\) of \(S\)-terms ordered by \(\sqsubseteq_\Ie\) and equipped with \(\rho : {SM(X, \le)} \to {M(X, \le)}\) given by \(\rho(\sigma, [p_1]_\Ie, \dots, [p_n]_\Ie) = [\sigma(\vec p)]_\Ie\). 
\end{definition}

The free \((S, \Ie)\)-algebra construction can easily be made into a functor \(M : \Pos \to \Pos\).
Given \(h : (X,\le) \to (Y, \le)\), define \(M(h) : M(X, \le) \to M(Y, \le)\) by setting\footnote{We often drop the notation \([~~]_\Ie\) and write \(p(\vec x)\) instead of \([p(\vec x)]_\Ie\).}
\(
	M(h)(x) = h(x)
\)
and
\(
	M(h)(\sigma(\vec p)) = \rho(\sigma, M(h)(p_1), \dots, M(h)(p_n))
\).
We also write \(\eta(x) = [x]_\Ie\) for each \(x \in X\).
Abstractly, the triple \((M, \eta, \rho)\) is an \emph{\((S,\Ie)\)-presented monad}, meaning that it satisfies
\[
	(*) \quad \parbox{0.95\textwidth}{
		For any \((S,\Ie)\)-algebra \((Y, \le, \alpha)\) and any \(h : (X, \le) \to (Y, \le)\), there is a unique \(S\)-algebra homomorphism \(h^\alpha : (M(X, \le), \rho) \to (Y, \le, \alpha)\) such that \(h^\alpha \circ \eta = h\).
	}
\]
There is at most one \((S, \Ie)\)-presented monad up to \(S\)-algebra isomorphism, so we will often say that \((M, \eta, \rho)\) is \emph{the} \((S, \Ie)\)-presented monad. 
In some situations, \(\eta\) and \(\rho\) are implicit, and then we simply say that \((S, \Ie)\) presents \(M\). 

\begin{assumption*}
	From here on, whenever \(v_1 \sqsubseteq_{\Ie} v_2\) with \(v_1,v_2 \in V\), we also have \(v_1\le_{V}v_2\). 
\end{assumption*}

This is equivalent to the assumption that \(\eta\) is an embedding, but it is also equivalent to the assumption that \(M(X, \le)\) can have more than one element. 
Our motivating examples are all obtained from equational theories.


\begin{example}\label{eg:guarded algebra}
	\emph{Ordered guarded algebra} is obtained from the equational theory of guarded semilattices in~\cite{schmidrozowskisilvarot2022processes}.
	For a fixed set of \emph{atomic tests} \(A\), it consists of the signature \(S_{\mathsf{ga}} = 1 + (2^\At, =) \times \Id^2\) and the inequational theory \(\theory{GA}\) made up of \(0 \le x\) and the equations
	\[
	x \qm_b x = x
	\qquad 
	x \qm_{\At} y = x
	\qquad 
	x \qm_b y = y \qm_{\bar b} x
	\qquad 
	(x \qm_b y) \qm_c z = x \qm_{bc} (y \qm_c z)
	\]
	where \(bc = b\cap c\) and \(\bar b = A \setminus b\). 
	Ordered guarded algebra presents \((\Id_\bot^\At, \lambda \alpha.(-), \qm)\), where 
	\[
		\lambda\alpha.x(\alpha) = x
		\qquad\qquad
		\qm(b, f_1, f_2)(\alpha) = (f_1 \qm_b f_2)(\alpha) = \begin{cases}
			f_1(\alpha) & \alpha \in b \\
			f_2(\alpha) & \alpha \notin b
		\end{cases}	
		\qquad 
		\qquad 
		(\alpha \in A)
	\]
	The functor part of this monad is the branching structure of ordered \acro{GKAT}-automata from \cref{eg:GKAT-automata}. 
	Guarded algebra originates in the work of McCarthy~\cite{mccarthy1961basis}, but has been reformed and rediscovered and renamed multiple times~\cite{bloomT1983ifthenelse,bloomE1988varieties,manes1991ifthenelse,gkat}.
\end{example}

\begin{example}\label{eg:convex algebra}
	\emph{Ordered convex algebra} is obtained from the equational theory of convex algebra~\cite{swirszcz1973convexity,flood1981semiconvex}. 
	This consists of the signature \((S_{\mathsf{ca}}, \theory{CA})\), where \(S_{\mathsf{ca}} = 1 + ([0,1], =) \times \Id^2\) and \(\theory{CA}\) contains \(0 \le x\) and the equations 
	\[
		x \oplus_r x = x
		\qquad 
		x \oplus_1 y = x
		\qquad 
		x \oplus_r y = y \oplus_{\bar r} x
		\qquad 
		(x \oplus_r y) \oplus_s z = x \oplus_{rs} (y \oplus_{\frac{\bar r s}{1-rs}} z)
	\]
	where \(\bar r = 1 - r\). 
	Ordered convex algebra presents\footnote{See \cref{app:ordered convex algebra}.} \((\D^\bot, \delta, \oplus)\) from \cref{eg:probabilistic processes}, where 
	\[
		\delta_x(y) = \begin{cases}
			1 & x = y \\
			0 & x \neq y
		\end{cases}	
		\qquad\qquad
		\oplus(r, \theta_1, \theta_2)(y) = (\theta_1 \oplus_r \theta_2)(y) = r~\theta_1(y) + (1-r)~\theta_2(y)
	\]
	This presentation of convex algebra in \(\Sets\) goes back to \'Swirszcz~\cite{swirszcz1973convexity} and Flood~\cite{flood1981semiconvex}.\footnote{Although it was likely known to Stone~\cite{stone1949barycentric}.}
\end{example}

\begin{example}
	\emph{The theory of ordered semilattices (with bottom)} is \((S_{\mathsf{sl}}, \theory{SL})\), where \(S_{\mathsf{sl}} = 1 + \Id^2\) and \(\theory{SL}\) consists of \(0 \le x\) and 
	\[
		0 \vee x = x
		\qquad
		x \vee x = x
		\qquad 
		x \vee y = y \vee x 
		\qquad 
		(x \vee y) \vee z = x \vee (y \vee z)
	\]
	It is well-known that the theory of ordered semilattices presents \((\A, \downarrow, \cup)\), where 
	\[
		\downset x = \{y \mid y \le x\}
		\qquad \qquad 
		\cup (U_1, U_2) = U_1 \cup U_2
	\]
	It can be derived from \(\theory{SL}\) that \(x \le y\) if and only if \(x = x \vee y\), so ordered semilattices are simply semilattices in the sense of order theory~\cite{birkhoff1940lattice}.
\end{example}

\subsection*{Iterative Inequational Theories}\label{sub:iterative inequational theories} 
Inequational theories give syntactic descriptions of ordered branching structures.
Such a syntax enables recursive specifications of ordered processes to be written as partially ordered systems of equations. 
For example, the process in \cref{eg:GKAT-automata} can be written 
\begin{equation}\label{eq:system eg 2}
	x_1 = \code ax_2 \qm_{\alpha_1} \code v 
	\qquad 
	x_2 = \code ax_2 \qm_{\alpha_1} \code v 
	\qquad
	x_3 = \code ax_2 \qm_{\alpha_1} 0
	\qquad \qquad 
	x_1 < x_2 > x_3
\end{equation}
This is an example of a \emph{guarded} system of equations, due to the fact that state variables only ever appear behind action symbols on the right-hand sides. 
Intuitively, (\ref{eq:system eg 2}) can safely be replaced by the following \emph{unguarded} system 
\[
	x_1 = \code ax_2 \qm_{\alpha_1} \code v	
	\qquad 
	x_2 = x_1 
	\qquad 
	x_3 = x_1 \qm_{\alpha_1} 0
	\qquad \qquad 
	x_1 < x_2 > x_3
\]
By inspection, the syntax of unguarded systems is less restricted, as it allows us to reuse state variables where we might have given explicit syntactic descriptions of transitions. 
However, the added freedom of allowing unguarded recursion in our specifications forces us to be explicit about the meanings of specifications such as \(x_1 = x_1 \qm_{\alpha_1} \code v\) and \(x_1 = x_1\).

We generalise Stark and Smolka's approach to unguarded recursion~\cite{stark1999complete}, which views these specifications as computations of certain least fixed points. 
We call inequational theories with least fixed points \emph{iterative}, a term we borrow from~\cite{elgot1975monadic,bloomE1976iterative,milius2005thesis}. 

\begin{definition}\label{def:iterative}
	An inequational theory \((S, \Ie)\) presenting a monad \((M, \eta, \rho)\) is \emph{iterative} if
	\begin{enumerate}
		\item for any poset \((X, \le)\) and any \(p(x, \vec y) \in S^*X\), the set
		\(
			\Pfp_x(p) = \{q \in S^*X \mid p(q, \vec y) \sqsubseteq_\Ie q\}
		\)
		has a  \(\sqsubseteq_\Ie\)-least element \(\lfp_xp\), called the \emph{least \(x\)-fixed point} of \(p\). 
		
		\item for any \(h : (X, \le) \to (Y, \le)\) and \((V, \le)\) with \(v \in V\), the map \[
		M(\id + h) : M((V, \le) + (X, \le)) \to M((V, \le)+ (Y, \le))
		\] 
		satisfies \[
		M(\id + h)(\lfp_v p(\vec u, \vec x)) = \lfp_v M(\id + h)(p)(\vec u, h(\vec x))
		\]
		for any term \(
			p(v, \vec u, \vec x) \in {M((V, \le) + (X, \le))}
		\)
		with \(v, \vec u \in V^m\) and \(\vec x \in X^n\).
	\end{enumerate}
\end{definition}

Recall that we think of an \(S\)-term as a representation of a branching configuration and the variables of an \(S\)-term as representatives of the branches that make up the configuration. 
In the spirit of the Tarski-Knaster fixed point theorem, condition {\textsf{1}} guaratees that every variable in a branching configuration has a least substitution producing the original configuration.
This is useful for abstractly representing the redistribution of branching probabilities.

\begin{restatable}{lemma}{lfpishonest}
	For any \(p(x, \vec y) \in S^*X\), the element \(\lfp_xp\) is the least fixed point of \(q \mapsto p(q, \vec y)\).
\end{restatable}

Condition {\textsf{2}} of \cref{def:iterative} says that irrelevant renamings of variables preserve the value of the least fixed point: this can be seen from the definition of \(M(\id + h)\), which has
\(
	M(\id + h)(p)(\vec u, h(\vec x)) = [p(\vec u, h(\vec x))]_{\Ie}
\).
Also note that every iterative inequational theory admits a constant \(0\) equivalent to \(\lfp_x x\), satisfying \(0 \sqsubseteq_\Ie p\) for all \(p \in S^*(X, \le)\). 

\begin{example}
	Ordered guarded algebra and the theory of ordered semilattices are both iterative. 
	In each theory, the least \(x\)-fixed point of the term \(p(x, \vec y)\) is simply \(p(0, \vec y)\). 
\end{example}

\begin{example}
	Ordered convex algebra is also an iterative inequational theory. 
	Given a subdistribution \(\theta\) on \((X, \le)\) of the form \(\theta = r~\delta_x + \theta'\), where \(\theta'(x) = 0\) and \(r < 1\), the least \(x\)-fixedpoint of \(\theta\) can be computed via conditioning: \(\lfp_x \theta = \theta'/(1-r)\).
	For \(r = 1\), the least \(x\)-fixedpoint is \(\delta_0\).
\end{example}

\section{A Family of Ordered Process Calculi}\label{sec:a family of ordered process calculi} 
Next, we introduce a parametrised family of calculi useful for specifying and reasoning about ordered processes. 
Each calculus is given by an iterative inequational theory \((S, \Ie)\), which determines the branching structure of the processes specified by terms of the language. 


For a given iterative inequational theory \((S, \Ie)\) that presents a monad \((M, \eta, \rho)\) on \(\Pos\), its poset of \emph{process terms} \((\Exp, \le_{\exp})\) has its underlying set given by 
\[
	\code v 
	\mid \sigma(e_1, \dots, e_n) 
	\mid \code ae 
	\mid \beta\code v~e 
	\mid \mu\code v~e
	\qquad\qquad 
	(\code v \in \Var, ~ \sigma \in I_n, ~ \code a \in \Act)
\]
and is ordered by the (\(\Sigma_M\)-)\emph{term ordering rules}, consisting of the \(S\)-term ordering rules and
\[
	\inferrule{\code a_1 \le_{\act} \code a_2}{\code a_1e \le_{\exp} \code a_2e} 
	\qquad 
	\inferrule{e_1 \le_{\exp} e_2}{\code ae_1 \le_{\exp} \code ae_2} 
	\qquad 
	\inferrule{e_1 \le_{\exp} e_2}{\beta\code v~e_1 \le_{\exp} \beta\code v~e_2} 
	\qquad 
	\inferrule{e_1 \le_{\exp} e_2}{\mu\code v~e_1 \le_{\exp} \mu\code v~e_2} 
\]
We say that a return variable \(\code v\) is \emph{bound} in a process term \(e\) if every instance of \(\code v\) appears within the scope of a \(\beta\code v\) or \(\mu\code v\), and otherwise we say that \(\code v\) is \emph{free} in \(e\). 
A return variable is said to be \emph{guarded} in \(e\) if it is either bound or appears within the scope of an \(\code a\). 

Abstractly, \((\Exp, \le_{\exp})\) carries the the \emph{initial} \(\Sigma_S\)-algebra \((\Exp, \le_{\exp}, \ev)\), where 
\[
	\Sigma_S = (\Var, =) + S + (\{\beta, \mu\}\times \Var, =) \times \Id 	
\]
and \(\ev : \Sigma_S(\Exp, \le_{\exp}) \to (\Exp, \le_{\exp})\) evaluates terms. 
This means that every \(\Sigma_S\)-algebra \((X, \le, \alpha)\) is reached by a unique algebra homomorphism \(\sem- : (\Exp, \le_{\exp}, \ev) \to (X, \le, \alpha)\). 

\subsection*{Small-step Semantics}\label{sub:small-step semantics} 
Process terms are specifications of processes. 
Intuitively, the return variable \(\code v\) denotes the process that immediately returns \(\code v\), \(\sigma(e_1, \dots, e_n)\) denotes the process that branches into the processes \(e_1\) through \(e_n\) with configuration \(\sigma\), and \(\code a e\) denotes the process that outputs \(\code a\) and moves on to the process \(e\).
In other words, the interpretations of the basic process term constructions are precisely those of the processes parametrised framework~\cite{schmidrozowskisilvarot2022processes}. 

The \emph{unguarded recursion operator} \(\beta\code v\), which does not appear in loc cit, computes the least fixed point with respect to unguarded instances of the return variable \(\code v\).
If we imagine a process specified by a process term as being rooted at that term, the operator \(\beta\code v\) leaves the state space untouched but reconfigures the outgoing transitions at the root.
This feeds into the interpretation of \(\mu \code v~e\), the \emph{guarded recursion operator} applied to \(e\), which denotes the process that runs \(\beta\code v~e\) until it encounters the return variable \(\code v\), at which point it repeats. 

Process terms are formally interpreted using a \emph{small-step semantics}, a monotonte system structure \((\Exp, \le_{\exp}, \epsilon)\) on process terms.
The structure map \(\epsilon\) is defined in \cref{fig:small-step}, although the semantics of the guarded recursion operator requires further explanation.
Given \(g \in \Exp\) and \(\code v \in \Var\), we define the \emph{guarded subtitution} operator \([g\gdsub{\code v}] : {B_M(\Exp, \le_{\exp})} \to B_M(\Exp, \le_{\exp})\) to be the unique \(S\)-algebra homomorphism such that 
\(
	\code u[g\gdsub{\code v}] = \code u
\)
and 
\( 
	\langle \code a, e\rangle[g\gdsub{\code v}] = \langle \code a, e[g\sub{\code v}]\rangle
\)
for any \(\code u \in \Var\), \(\code a \in \Act\), and \(e \in \Exp\), where \(e[g\sub{\code v}]\) is the process term obtained by substituting each free instance of \(\code v\) with \(g\). 

\begin{figure}[t!]
	\centering 
	\begin{mathpar}
		\epsilon(\code v) = \code v
		\and 
		\epsilon(\sigma(e_1, \dots, e_n))
		= \sigma(\epsilon(e_1), \dots, \epsilon(e_n))
		\and
		\epsilon(\code ae) = \langle \code a, e\rangle 
		\\
		\epsilon(\beta \code v~e) = \lfp_{\code v} \epsilon(e)
		\and
		\epsilon(\mu \code v~e) = \lfp_{\code v} \epsilon(e)[\mu\code v~e\gdsub \code v]
	\end{mathpar}
	\caption{The small-step semantics \(\epsilon : (\Exp, \le_{\exp}) \to B_M(\Exp, \le_{\exp})\) of process expressions.\label{fig:small-step}}
\end{figure}

\begin{example}
	Consider the process term \(e_1 = \mu\code u~\code a (\code a\code u \qm_b \code v)\) and let \(e_2 = \code ae_1 \qm_b \code v\). 
	The ordered \acro{GKAT}-automaton specified by \(e_1\) is the discrete automaton
	
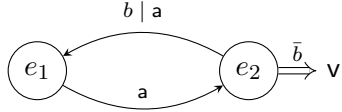
\[\begin{tikzpicture}
		\node[state] (0) {\(e_1\)};
		\node[state, right=2cm of 0] (1) {\(e_2\)};
		\node[right=0.5cm of 1] (2) {\(\code v\)};
		\draw (0) edge[bend right] node[above] {\(\code a\)} (1);
		\draw (1) edge[bend right] node[above] {\(b\mid \code a\)} (0);
		\draw (1) edge[-implies, double, double distance=2pt] node[above] {\(\bar b\)} (2);
	\end{tikzpicture}\] 
\end{example}

\section{Final/Initial Semantics}\label{sec:semantics} 
The small-step semantics of a process term allows us to visualise its execution an \(M\)-branching labelled transition system through state diagrams like the one in \cref{eg:GKAT-automata}. 
By unrolling, we are able to compare process terms for behavioural equivalence. 
For example, the terms \(\mu \code u~\code a \code u\) and \(\mu\code u~\code a \code a \code u\) both denote systems that endlessly repeat the action \(\code a\).
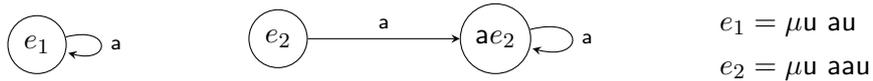
\[\begin{tikzpicture}
	\node[state] (0) {\(e_1\)};
	\draw[loop] (0) edge[->, loop right] node[right] {\(\code a\)} (0);
\end{tikzpicture}
\qquad\qquad
\begin{tikzpicture}
	\node[state] (0) {\(e_2\)};
	\node[state, right=2cm of 0] (1) {\(\code ae_2\)};
	\draw (0) edge[->] node[above] {\(\code a\)} (1);
	\draw[loop] (1) edge[->, loop right] node[right] {\(\code a\)} (1);
\end{tikzpicture}
\qquad\qquad
\begin{aligned}[b]
	e_1 &= \mu \code u~\code a \code u	\\
	e_2 &= \mu \code u~\code a\code a \code u	
\end{aligned}\]
The first system has one state and the second has two, but they unroll to the same behaviour. 
The  map \(\beh_\epsilon\) provides the \emph{final} or \emph{operational semantics} of process terms~\cite{turiplotkin1997operational}. 

\subsection*{Denotational Semantics}\label{sub:denotational semantics} 
The \(\Sigma_S\) operations also correspond to sensible constructions in \((Z, \le_\beh)\). 
For example, the behaviours of \acro{GKAT}-automata are \(A\)-branching trees~\cite{schmidkappekozensilva2021gkat}, and given \(t_1, t_2 \in Z\), \(t_1 \qm_b t_2\) denotes the tree with the same \(b\)-branches as \(t_1\) and the same \(\bar b\)-branches as \(t_2\) for each test \(b \subseteq A\). 
Formally, the behaviours \(\sigma(t_1, \dots, t_n)\), \(\code a t\), \(\beta\code v~t\), and \(\mu\code v~t\) are defined in terms of behavioural differential equations in the sense of~\cite{rutten1998coinduction}.
The \(\Sigma_S\)-algebra structure \((Z, \le_\beh, \gamma)\) is summarised in \cref{fig:behaviour algebra}, but again the definition of \(\mu\code v~t\) requires explanation.

\begin{figure}
	\centering 
	\begin{mathpar}
		\zeta(\code v) = \code v 
		\and 
		\zeta(\gamma(\sigma, t_1, \dots, t_n)) = \sigma(\zeta(t_1), \dots, \zeta(t_n))
		\and
		\zeta(\zeta(\code a, t)) = \langle \code a, t\rangle \\
		\zeta(\gamma(\beta\code v, t)) = \lfp_{\code v} \zeta(t) 
		\and 
		\zeta(\gamma(\mu\code v, t)) = \lfp_{\code v} \zeta(t)\{\gamma(\mu\code v, t)\gdsub{\code v}\} 
	\end{mathpar}
	\caption{The operation map \(\gamma : \Sigma_S(Z, \le_\beh) \to (Z, \le_\beh)\) of the \(\Sigma_S\)-algebra \((Z, \le_\beh, \gamma)\).\label{fig:behaviour algebra}}
\end{figure}

The behavioural differential equation defining the behaviour \(\mu \code v~t\) involves \emph{semantic} and \emph{guarded semantic substitution}, analogues of syntactic and guarded sytactic substitution in \(Z\).
For any \(t,s \in Z\), and any \(\code v \in \Var\), we also define \(t\{s\sub\code v\}\) with a behavioural differential equation: suppose \(\zeta(t) = p(\code v, \vec {\code u}, \langle \code a_i, t_i\rangle_{i \le n})\) for some \(p(v, \vec u, \vec z)\). 
Then
\[
	\zeta(t\{s\sub\code v\}) = p(\zeta(s), \vec {\code u}, \langle \code a_i, t_i\{s\sub\code v\}\rangle_{i \le n})
\qquad 
\text{and} 
\qquad
	\zeta(t)\{s \gdsub \code v\} = p(\code v, \vec{\code u}, \langle \code a_i, t_i\{s\sub\code v\}\rangle_{i \le n})
\]
Intuitively, \(t\{s\sub\code v\}\) is obtained from the behaviour \(t\) by replacing every return statement (returning \(\code v\)) with the behaviour of \(s\), and \(\{s\gdsub\code v\}\) replaces only the return statements that occur after an action has been performed. 

This allows for another approach to comparing process terms for behavioural equivalence, using the initiality of the \(\Sigma_S\)-algebra \((\Exp, \le_{\exp}, \ev)\).
The unique \(\Sigma\)-algebra homomorphism \(\sem- : (\Exp, \le_{\exp}, \ev) \to (Z, \le_\beh, \gamma)\) computes the \emph{initial} or \emph{denotational semantics} of process terms. 
The following theorem, which states that the initial and final semantics coincide, implies that the behaviours of process terms form a subalgebra of \((Z, \le_\beh, \gamma)\).

\begin{restatable}{theorem}{adequacytheorem}\label{thm:adequacy}
	For any \(e \in \Exp\), \(\sem e = \beh_\epsilon(e)\). 
\end{restatable}

Given an ordered algebra \((X, \le, \alpha)\), a preorder \(\preceq\) on \(X\) is called a \emph{precongruence} on \((X, \le, \alpha)\) if it extends \(\le\) and is preserved by \(\alpha\). 
Precongruences on \((X, \le, \alpha)\) are in one-to-one correspondence with quotient algebras of \((X, \le, \alpha)\).
In one direction, if \(h : {(X, \le, \alpha)} \to {(Y, \le, \alpha')}\), we define \(x \preceq y\) in \(X\) if and only if \(h(x) \le h(y)\).
In the other, given a precongruence \(\preceq\) on \({(X, \le, \alpha)}\) we form its quotient \({(X/{\approx}, \preceq, \bar \alpha)}\), where \(x \approx y\) if and only if \(x \preceq y \preceq x\)~\cite{kapulkinkurzvelebil2012coalgposet}.

If we write \(J\) for the image of \(\Exp\) under \(\sem-\), \cref{thm:adequacy} implies that the behaviour order is a precongruence on \((\Exp, \sqsubseteq, \ev)\), and consequently \(\le_\beh\) and \(\gamma\) restrict to \(J\). 
This determines a unique subalgebra structure \({(J, \le_\beh, \gamma_J)} \hookrightarrow {(Z, \le_\beh, \gamma)}\).
In the next section, we use \(\Ie\) to axiomatise \((J, \le_\beh, \gamma_J)\). 

\section{Axiomatisation}\label{sec:axiomatisation} 
In this section, we add an axiomatisation of the behaviour order to the framework.

\begin{definition}
	The \emph{axiomatic order} is the smallest precongruence \(\sqsubseteq\) on \((\Exp, \le_{\exp}, \ev)\) containing the pairs \(e \sqsubseteq f\) derivable from the rules in \cref{fig:axioms}.	
\end{definition}

We write \(e \equiv f\) to denote that \(e \sqsubseteq f \sqsubseteq e\) and define \(\E = \Exp/{\equiv}\).
The equivalence classes in \(\E\) are called \emph{term behaviours}, and term behaviours are partially ordered by the relation \(\sqsubseteq\) defined so that \([e]_\equiv \sqsubseteq [f]_\equiv\) iff \(e \sqsubseteq f\).
Since \(\sqsubseteq\) is a precongruence on \((\Exp, \le_{\exp}, \ev)\), term behaviours form a \(\Sigma_S\)-algebra \((\E, \sqsubseteq, \ev)\).

\begin{figure}[!t]
	\begin{mathpar}
		\prftree[l]{(\(\Ie\))}
		{e \sqsubseteq_\Ie f}
		{e \sqsubseteq f}
		\and
		\prftree[l]{(\textsf{S})}
		{\sigma_1 \le_n \sigma_2}
		{\sigma_1(\vec e) \sqsubseteq \sigma_2(\vec e)}
		\and
		\prftree[l]{(\textsf{Act})}
		{\code a_1 \le_{\act} \code a_2}
		{\code a_1 e \sqsubseteq \code a_2 e}
		\\
		\prftree[l]{(\textsf{R1a})}
		{(\forall i)~\code v \in \gd(f_i)}
		{\beta\code v~p(\code v,\vec f) \equiv (\lfp_{\code v}~p)(\vec f)}
		\and
		\prftree[l]{(\textsf{R1b})}
		{\ }
		{(\beta \code v~e)[\mu \code v~e\sub \code v] \sqsubseteq \mu \code v~e}
		\and
		\prftree[l]{(\textsf{R2a})}
		{p(g, \vec f) \sqsubseteq g \and (\forall i)~\code v \in \gd(f_i)}
		{\beta \code v~p(\code v, \vec f) \sqsubseteq g}
		\and 
		\prftree[l]{(\textsf{R2b})}
		{(\beta\code v~ e)[g\sub\code v] \sqsubseteq g}
		{\mu \code v~e \sqsubseteq g}
		\and
		\prftree[l]{(\textsf{R3})}
		{g \sqsubseteq (\beta\code v~e)[g \sub \code v]}
		{g \sqsubseteq \mu \code v~e}	
	\end{mathpar}
	\caption{
		Axioms for \(\sqsubseteq\). 
		Here, \(e,f,g \in \Exp\) and \(\vec e = (e_1,\dots,e_n),\vec f = (f_1, \dots, f_n) \in \Exp^n\), and \(\code a, \code a_1, \code a_2 \in \Act\), and \(\sigma,\sigma_1,\sigma_2,\sigma \in I_n\), and \(\code v \in \Var\). 
		We also take \(p(\code v, \vec x) \in S^*(\Var + X)\).  \label{fig:axioms}
	}
\end{figure}

Given a preorder \(\preceq\) extending the partial order on a poset \((X, \le)\), and given a monotone map \(h : X \to D\) into another poset \((D, \le)\), we say that \(\preceq\) is \emph{sound} with respect to \(h\) if \(x \preceq y\) implies \(h(x) \le h(y)\) for all \(x,y \in X\).
We say that \(\preceq\) is \emph{complete} with respect to \(h\) if, conversely, \(h(x) \le h(y)\) implies \(x \preceq y\).
Soundness is a consequence of the following theorem. 

\begin{restatable}{theorem}{intermediatesoundnesstheorem}
	The quotient map \([-]_{\equiv} : \Exp \to \E\) is a monotone coalgebra homomorphism \((\Exp, \le_{\exp}, \epsilon) \to (\E, \sqsubseteq, \bar\epsilon)\) for a unique coalgebra structure \(\bar\epsilon : (\E, \sqsubseteq) \to B_M(\E, \sqsubseteq)\). 
\end{restatable}

\begin{corollary}[Soundness]
	Let \(e,f \in \Exp\).
	If \(e \sqsubseteq f\), then \(e \le_\beh f\).
\end{corollary}

\begin{proof}
	Suppose \(e \sqsubseteq f\).
	Then \([e]_\equiv \sqsubseteq [f]_\equiv\) in \(\E\). 
	Since the behaviour map is monotone, \([e]_\equiv \le_\beh [f]_\equiv\).
	From uniqueness of the unrolling map, \(\beh_\epsilon = \beh_{\bar\epsilon} \circ [-]_\equiv\).
	It follows that \(e \le_\beh f\).
\end{proof}

We now turn our attention to completeness.
So far, we have the unrolling map \(\beh_\epsilon = \sem- : (\Exp, \le_{\exp}) \to (J, \le_\beh)\) taking each process to its behaviour. 
Completeness requires that \(\sem-\) is an embedding up to \(\equiv\).
It suffices to find a monotone left inverse \(\phi\) of \(\sem-\) up to \(\equiv\), as the existence of such a map implies
\(
e \equiv \phi(\sem e) \sqsubseteq \phi(\sem f) \equiv f
\)
whenever \(\sem e \le_\beh \sem f\). 
Our construction of \(\phi\) follows~\cite{salomaa1966two,kozen1991completeness,jacobs2006bialgebraic,silva2010kleene}\footnote{A summary of the exact method that appears in these texts is~\cite[Theorem 5.6]{schmid2021star}. Technically, we use a variant of this theorem for ordered processes.}, and consists of the four steps below:
\begin{enumerate}
	\item Identify ordered processes with systems of equations (see for {eg.} (\ref{eq:system eg 2})).
	\item Show that their solutions coincide with coalgebra homomorphisms into \((\E, \sqsubseteq, \bar\epsilon)\).
	\item Construct a unique solution to every finite system of equations.
	\item Exhibit the behaviour of each process term with a finite ordered process. 
\end{enumerate} 
Then, given \(t \in J\), the construction of \(\phi(t)\) is carried out by finding a finite subcoalgebra \(U\) of \((J, \le_\beh, \zeta_J)\) containing \(t\) and solving its associated system of equations. 


\begin{theorem}
	The map \(\phi : (J, \le_\beh) \to (\E, \sqsubseteq)\) is left-inverse to \(\sem-\) up to \(\equiv\).
\end{theorem}

\begin{corollary}[Completeness]\label{thm:completeness}
	For any \(e,f \in \Exp\), if \(e \le_{\beh} f\), then \(e \sqsubseteq f\).
\end{corollary}

\begin{example}
	The process calculus obtained from the theory of convex algebra in \cref{eg:guarded algebra} has the syntax 
	\[
		\code v 
		\mid e_1 \oplus_r e_2  
		\mid \code ae 
		\mid \beta\code v~e 
		\mid \mu\code v~e
		\qquad\qquad 
		(\code v \in \Var, ~ r \in [0,1], ~ \code a \in \Act)
	\]
	The operation \(\oplus_r\) can be interpreted as a weighted coin flip. 
	This gives a calculus that is equally as expressive as Stark and Smolka's \acro{APA}~\cite{stark1999complete}, and furthermore the small-step semantics \(\epsilon(\mu \code v~e)\) coincides with Stark and Smolka's.
	As we will see in \cref{sec:simulations}, discrete ordered probabilistic processes are behaviourally equivalent if and only if they are behaviourally equivalent in the coalgebraic sense.
	Therefore, \cref{thm:completeness} is a generalisation of Stark and Smolka's compeleteness theorem~\cite[Theorem 3]{stark1999complete}.
\end{example}

\section{Star Fragments}\label{sec:star fragments}
%
One of the most striking features of the processes parametrised framework~\cite{schmidrozowskisilvarot2022processes} is that it produces a calculus of terms resembling regular expressions for every equational theory consisting of binary operations and a deadlock symbol. 
Given a signature of the form \(S = 1 + I \times \Id^2\), we form its \emph{star fragment} \(\SExp\) with the grammar
\[
	0 \mid 1 \mid \code a \mid e +_\sigma f \mid ef \mid e^{(\sigma)} 
	\qquad\qquad
	(\code a \in \Act,~\sigma \in I,~e,f \in \SExp)
\]
Elements of \(\SExp\) are called \emph{star expressions}, and can be interpreted as process terms via a translation \(\tau : \SExp \to \Exp\) involving a pair of fixed return variables \(\unit, \code v \in \Var\).
The variable \(\unit\) is called the \emph{unit}, and is intended to represent successful termination.
The translation is defined inductively as follows: \(\tau(0) = 0\), \(\tau(1) = \unit\), \(\tau(\code a) = \code a\unit\), and
\[
	\tau(e +_\sigma f) = \sigma(\tau(e), \tau(f))
	\qquad 
	\tau(ef) = \tau(e)[\tau(f)\sub\unit]
	\qquad 
	\tau(e^{(\sigma)}) = \mu \code v~\sigma(\tau(e)[\code v\sub\unit], \unit)
\]
for any \(\code a \in \Act\) and \(\sigma \in I\). 
Intuitively, \(+_\sigma\) represents the binary operation \(\sigma\), action symbols represent themelves followed by successful termination, and \(e^{(\sigma)}\) iterates both \(e\) in the left branch and successful termination in the right branch.

Milner is the first to observe a star fragment in~\cite{milner1984complete}, corresponding to the theory of semilattices with bottom. 
His star expressions are commonly known as \emph{regular expressions modulo bisimilarity}, because their syntax is identical to that of Kleene algebra but their semantics is nondeterministically branching trees rather than languages.
\acro{GKAT}~\cite{gkat} is also an example: it is the star fragment for guarded algebra\footnote{Called the theory of \emph{guarded semilattices} in loc cit. We suggest a change in name to be consistent with convex algebra and convex semilattices, only the latter of which contains the semilattice equations.}~\cite{schmidrozowskisilvarot2022processes}. 

The theory of pointed convex algebra suggests another example of a star fragment in the processes parametrised framework, with Minkowski sums \(e \oplus_r f\) as its binary operations and Bernoulli-like processes \(e^{(r)}\) as its loops.
However, the loops in the language satisfy the peculiar equation \((e\oplus_s 1)^{(r)} = (e \oplus_s 0)^{(r)}\), which is not sound with respect to Stark and Smolka's interpretation of the probabilistic loop (explained in the introduction)~\cite{stark1999complete}.
Again, this is easily mended by interpreting star expressions for pointed convex algebra in the ordered processes parametrised framework instead. 

Star expressions are equally well-defined in the ordered processes parametrised framework. 
Star expressions form the poset \((\SExp, \le_{\text{sxp}})\) with a term order, similar to \(\le_{\exp}\), such that the translation map \(\tau\) is monotone. 
As we see from the probabilistic case, the \emph{tightening axiom} \((e +_{\sigma_2} 1)^{(\sigma_1)} \equiv (e +_{\sigma_2} 0)^{(\sigma_1)}\) proposed in~\cite{schmidrozowskisilvarot2022processes} is not always sound in the ordered case. 
This is mended for probabilistic star expressions by replacing the tightening axiom with 
\begin{equation}\label{eq:tightening for probabilities}
	(e \oplus_s 1)^{(r)} = e^{(\frac{rs}{1-r\bar s})}
\end{equation}
The reason is that the small-step semantics of the expression 
\begin{equation}\label{eq:eg redistribute}
	\tau((e \oplus_s 1)^{(r)}) 
	= \mu \code v~((\tau(e)[\code v\sub\code u] \oplus_s \code v) \oplus_r \unit)
	\equiv \mu \code v~(\tau(e)[\code v\sub\code u] \oplus_{sr} (\code v \oplus_{\frac{r\bar s}{1 - rs}} \unit))
\end{equation}
is computed by first determining the least \(\code v\)-fixed point of the scope of the \(\mu \code v\) to the right of (\ref{eq:eg redistribute}), which redistributes the probability \(1-r\bar s\) to the other terms\footnote{With some algebra, one can derive \(\oplus_{rs/(1-r\bar s)}\) as the correct loop guard.}. 

\begin{example}\label{eg:ProbGKAT}
	Consider the inequational theory \((S_{\mathsf{pg}}, \theory{PG})\) with \(S = 1 + (2^A + [0,1])\times \Id^2\) and \(\theory{PG}\) consisting of the guarded algebra axioms from \cref{eg:guarded algebra}, the convex algebra axioms from \cref{eg:convex algebra}, and the distribution axiom 
	\[x \oplus_r(y \qm_b z) = (x \oplus_r y) \qm_b (x \oplus_r z)\]
	The monad presented by \((S_{\mathsf{pg}}, \theory{PG})\) is carried by \((\D^\bot(\Id))^A\), same as for \acro{ProbGKAT}~\cite{rozowskikappekozenschmidsilva2022probgkat}.
	The \acro{ProbGKAT} syntax, small-step semantics, and axiomatisation are precisely those obtained from the star fragment of the process calculus for \((S_{\mathsf{pg}}, \theory{PG})\), so long as we 
	\begin{itemize}
		\item add symbols for constants \(\code v \in \Var\) to the grammar generating \(\SExp\),
		\item add the axioms \(\code v e = \code v\) for each \(\code v \in \Var\), 
		\item and replace the probabilistic tightening axiom from~\cite{schmidrozowskisilvarot2022processes} with (\refeq{eq:tightening for probabilities}).
	\end{itemize}
\end{example}

It is not clear how to adapt (\refeq{eq:tightening for probabilities}) to other star fragments of ordered process calculi, since resolving the least \(\code v\)-fixed point of \(e^{(\sigma)}\) can produce an expression outside of the star fragment.
Thus, we introduce a larger fragment, denoted \(\PExp\) and called the \emph{polystar fragment} of \(\Exp\).
The elements of \(\PExp\), \emph{polystar expressions}, are generated by the grammar
\begin{equation}\label{eq:polystar grammar}
	0 \mid 1 \mid \code a \mid e +_\sigma f \mid ef \mid e^{[p]} 
	\qquad\qquad
	(p(x,y) \in S^*\{x,y\})
\end{equation}
Polystar expressions are interpreted as process terms via a translation function \(\tau : \PExp \to \Exp\) defined inductively: \(\tau(0) = 0\), \(\tau(1) = \unit\), \(\tau(\code a) = \code a\unit\), and
\[
	\tau(e +_\sigma f) = \sigma(\tau(e), \tau(f))
	\qquad 
	\tau(ef) = \tau(e)[\tau(f)\sub\unit]
	\qquad 
	\tau(e^{[p]}) = \mu \code v~p(\tau(e)[\code v\sub\unit], \unit)
\]
for any \(\code a \in \Act\), \(\sigma \in I\), and \(p(x, y) \in S^*\{x,y\}\).
Sound axioms for the behaviour order on polystar expressions include \(0 \le e\), \(0 e = 0 \), \(1e = e = e1\), \(e_1 (e_2 e_3) = (e_1 e_2) e_3\), and
\begin{gather*}\label{fig:axioms for stars} 
	p(e e^{[p]}, 1) \le e^{[p]}
	\qquad
	(e +_\sigma 1)^{[p]} = e^{[\lfp_{v}p(x +_\sigma v, y)]}
	\\
	\inferrule
	{e \sqsubseteq_\Ie f}
	{e \le f}
	\quad
	\inferrule
	{\sigma_1 \le_S \sigma_2}
	{e +_{\sigma_1} f \le e +_{\sigma_2} f}
	\quad
	\inferrule
	{\code a_1 \le_{\act} \code a_2}
	{\code a_1 \le \code a_2}
	\qquad 
	\inferrule 
	{g \le p(eg, 1) \qquad e~\text{ is guarded}}{g \le e^{[p]}}
\end{gather*}
where a generalised version of (\refeq{eq:tightening for probabilities}) appears.

We have two pressing questions about polystar fragments. Firstly, are the axioms listed above complete with respect to behavioural equivalences between polystar expressions?
This question is inspired by two very difficult problems:  one posed by Milner in~\cite{milner1984complete} asking whether an axiomatisation inspired by Salomaa's first axiomatisation of Kleene algebra~\cite{salomaa1966two} is complete for regular expressions up to bisimulation\footnote{After 38 years of being open, a positive answer has recently been announced by Clemens Grabmeyer~\cite{grabmayer2022complete}!}, and the question of whether or not the analogous axiomatisation for \acro{GKAT} is complete with respect to tree-unrollings~\cite{schmidkappekozensilva2021gkat}.

Our second question has to do with expressivity. 
By setting \(p(x, y) = x +_\sigma y\) in (\ref{eq:polystar grammar}) we see that every star expression exhibits a behaviour that can be specified with a polystar expression. 
For which algebraic theories is the converse also true?
We have managed to show that this is the case for convex algebra, guarded algebra, and the theory of semilattices with bottom, but not for the theory of guarded convex algebra from \cref{eg:ProbGKAT}. 

\section{On Behavioural Equivalences}\label{sec:simulations}
In this section, we discuss the relationship between behavioural equivalence for ordinary systems, the behaviour order for monotone systems, and the notion of similarity from process theory through a careful analysis of \((S, \Ie)\). 

For the duration of this section, we assume that the action order \(\le_{\act}\) is discrete, so that the underlying partial order on states is the only external influence on the order of \(B_M\).  
We let \(U : \Pos \to \Sets\) denote the forgetful functor \(U(X, \le) = X\) and \(D : \Sets \to \Pos\) denote the discrete poset functor \(DX = (X, =)\) and recall that \(UD = \Id_{\Sets}\).
We also fix an inequational theory (not necessarily iterative) \((S, \Ie)\) presenting the monad \((M, \eta, \rho)\) and set \(N = UMD\) and \(B_N = N(\Var + \Act \times \Id_{\Sets})\).

We refer to a pair \((X, \vartheta)\) with \(\vartheta : X \to B_N X\) as an \emph{ordinary system}, and write \((X, =, \vartheta)\) for the discrete system with \(\vartheta : DX \to B_MDX\).
Since both \((X, \vartheta)\) and \((X, =, \vartheta)\) have \(X\) as its underlying set, two notions of behavioural equivalence are present for elements of \(X\): \emph{ordered behavioural equivalence} \(x =_\beh y\) from \cref{sec:ordered processes}, and \emph{ordinary behavioural equivalence} \(x \bisim y\), which holds if \(h(x)= h(y)\) for some \(h : (X, \vartheta) \to (Y, \delta)\).

\begin{restatable}{lemma}{ordinaryimpliesordered}\label{lem:ordinary implies ordered}
	Let \((X, \vartheta)\) be an ordinary process with states \(x\) and \(y\).
	If \(x \bisim y\), then \(x =_\beh y\).
\end{restatable}

To illustrate the difference between these notions of behavioural equivalence, consider \cref{eg:labelled transition systems}, where ordered labelled transition systems are obtained by setting \(M = \A\) and (ordinary) labelled transition systems are obtained from \(N = \P\). 
For any \(\code a \in \Act\) and \(\code v \in \Var\), we have \(\code {a}0 + \code {aav} =_\beh \code {aa v}\), but it is not true that \(\code {a}0 + \code {aav} \bisim \code{aav}\).
Characterising the situations in which ordered behavioural equivalence implies ordiary behavioural equivalence is straightforward: this occurs if and only if the equations derived from \((S, \Ie)\) present \(N\) (in the sense of~\cite{bonchiSV2019tracesfor}). 
Recall that we say \(M\) \emph{lifts} \(N\) if \(UM = NU\). 

\begin{restatable}{theorem}{orderedimpliesordinary}\label{lem:ordered implies ordinary}
	Ordered behavioural equivalence coincides with ordinary behavioural equivalence if and only if \(M\) lifts \(N\). 
\end{restatable}

In many cases, the behaviour order is contained in the \emph{simulation order} for ordinary processes, a coinductive predicate that witnesses a form of behavioural inclusion~\cite{hughesjacobs2004simulations}.
\begin{definition}
	Given ordered processes \((X, \le, \vartheta)\) and \((Y, \le, \delta)\), a relation \(R \subseteq X \times Y\) is called a \emph{simulation} if for any \((x, y) \in R\), there is a \(p(\vec{\code v}, \langle \code a_i, (x_i, y_i)\rangle_{i \le n}) \in B_MR\) such that 
\[
	\vartheta(x) \sqsubseteq_\Ie p(\vec{\code v}, \langle \code a_i, x_i\rangle_{i \le n})
	\qquad\text{and}\qquad 
	 p(\vec{\code v}, \langle \code a_i, y_i\rangle_{i \le n}) \sqsubseteq_\Ie \delta(y) 
\]
We write \(x \precsim y\) and say that \(x\) \emph{is similar to} \(y\) if \((x,y)\) is contained in a simulation, and that \(x\) and \(y\) are \emph{two-way similar} if \(x \precsim y\) and \(y \precsim x\). 
\end{definition}

For example, simulations between LTSs are concretely characterised as follows: a relation \(R \subseteq X \times Y\) between LTSs \((X, \vartheta)\) and \((Y, \delta)\) is a simulation if and only if for any \((x, y) \in R\), (i) if \(x \Rightarrow \code v\), then \(y \Rightarrow \code v\), and (ii) if \(x \tr{\code a} x'\), then \((\exists y')~(x',y') \in R\) and \(y \tr{\code a} y'\).

Ordinary behavioural equivalence is not implied by two-way similarity for LTSs~\cite{baetenweijland1990processalgebra}. 
As it so happens, the behaviour order \(\le_\beh\) on a discretely ordered LTS \((X, =, \vartheta)\) is a simulation on \((X, \vartheta)\), so provable inequivalence is sound with respect to similarity.
This is as far as the connection goes between ordered behavioural equivalence and two-way similarity, however: the processes denoted \(\code a_1\code a_2\code v\) and \(\code a_1(\code a_2 \code v+ \code a_3\code v)\) are two-way similar but not ordered behaviourally equivalent since \(\code a_2 \code v+ \code a_3\code v \le_\beh \code a_2 \code v\) does not hold. 
In sum, 
\begin{equation}\label{eq:imps for LTSs}
	{x \bisim y} 
	\stackrel{(\text a)}\implies {x =_\beh y} 
	\stackrel{(\text b)}\implies {x \sim y} 
\end{equation}
for LTSs, but neither converse holds in general.

By \cref{lem:ordinary implies ordered}, (\refeq{eq:imps for LTSs}.a) is always true. 
A characterisation of the situations in which (\refeq{eq:imps for LTSs}.b) holds can be given in terms of the inequational theory \((S, \Ie)\). 

\begin{definition}\label{def:weak coupling property}
	We say that \((S, \Ie)\) \emph{admits weak couplings} if for any \(h : (X, \le) \to (Y, \le)\) and \(p(\vec x), q(\vec y) \in M(X, \le)\) such that \(p(h(\vec x)) \sqsubseteq_\Ie q(h(\vec y))\), there is a term \(r((u_1,v_1), \dots, (u_n,v_n)) \in M(X,\le)^2\) such that \(p \sqsubseteq_\Ie r(\vec u)\) and \(r(\vec v) \sqsubseteq_\Ie q\) and \(h(u_i) \le h(v_i)\) for \(i=1,2\).
	In such a case, we call the term \(r\) a \emph{weak \(h\)-coupling}\footnote{We borrow the name from probability theory~\cite{thorisson1995coupling}.} of \(p\) with \(q\).
\end{definition}

Intuitively, an inequational theory admits weak couplings if every inequation introduced by some monotone substitution originates in a precongruence. 
At the level of processes, precongruences play the role of simulations.
This suggests that the admition of weak couplings implies that the behaviour order is a simulation.

\begin{restatable}{lemma}{behaviouralinequivalenceisasim}\label{lem:behavioural inequivalence is a simulation}
	Suppose that \((S, \Ie)\) admits weak couplings and let \((X, \vartheta)\) be a \(B_N\)-coalgebra. 
	Then \(\le_\beh\) is a simulation on \((X, \vartheta)\).
\end{restatable}


We also say that \((S, \Ie)\) \emph{admits couplings} if for any \(h : (X, \le) \to (Y, \le)\) and any \(p(\vec x),q(\vec y) \in M(X, \le)\) such that \(p(h(\vec x)) \equiv_\Ie q(h(\vec y))\), there is an \(h\)-coupling \(r\) of \(p\) and \(q\). 
The term \(r\) is called an \emph{\(h\)-coupling} in such a case. 

\begin{remark}\label{rem:admits couplings wpb}
	Admitting couplings is equivalent to \(M\) preserving weak pullbacks, a common assumption in the structure theory of coalgebra~\cite{gumm1998functorsforcoalgebras}. 
\end{remark}

Finally, we give a sufficient condition for the converse of (\refeq{eq:imps for LTSs}) to hold in general. 
We name the condition in opposition with the type of branching of LTSs, nondeterminism. 

\begin{definition}
	We say that \((S, \Ie)\) is \emph{deterministic} if for any \(h : (X, \le) \to (Y, \le)\) a pair \(p, q \in M(X, \le)\) admits an \(h\)-coupling if and only if there is a weak \(h\)-coupling of \(p\) with \(q\) and a weak \(h\)-coupling of \(q\) with \(p\).
	Otherwise, we say that \((S,\Ie)\) is \emph{nondeterministic}. 
\end{definition}

 
\begin{restatable}{theorem}{collapsetheorem}\label{thm:collapse}
	Suppose \((S, \Ie)\) is deterministic, admits weak couplings, and \(M\) lifts \(N\).
	If \(x\) and \(y\) are states of a \(B_N\)-coalgebra, then \(x \bisim y\) if and only if \(x =_\beh y\) if and only if \(x \sim y\). 
\end{restatable}

\begin{example}
	Both the inequational theories of guarded algebra and convex algebra are deterministic and admit weak couplings. 
	By \cref{thm:collapse}, two-way simulation and behavioural equivalence coincide for \acro{GKAT}-automata and ordered probabilistic processes.
	This relationship is well-known for probabilistic processes~\cite{desharnais1999logical}, but is perhaps interesting to derive from a universal algebra point of view. 
\end{example}

\begin{example}
	By \cref{thm:collapse}, either the inequational theory of semilattices with bottom fails to admit weak couplings or fails to be deterministic. 
	Weak couplings are easy to spot, however: if \(p = \downset\{x_1, \dots, x_n\}\) and \(q = \downset\{y_1, \dots, y_m\}\) in \(\A(X, \le)\) and \(h : (X, \le) \to (Y, \le)\) satisfies 
	\(
		\downset\{h(x_1), \dots, h(x_n)\} \subseteq \downset\{h(y_1), \dots, h(y_m)\}
	\)
	then let \((u_i, v_i)_{i \le k}\) denote all of the pairs \((x_i, y_j)\) such that \(h(x_i) \le h(y_j)\) for \(i \le n\) and \(j \le m\). 
	Then \(\downset\{(u_1,v_1), \dots, (u_n,v_n)\}\) is a weak \(h\)-coupling of \(p\) with \(q\). 
	Consequently, \((S_{\mathsf{sl}}, \theory{SL})\) is \emph{non}deterministic.
\end{example}

\section{Discussion}\label{sec:discussion} 
This paper is inspired by the body of work that begins with the process calculi of Milner~\cite{milner1980ccs} and Stark and Smolka~\cite{stark1999complete} and finds its way to general specification systems for coalgebras in \(\Sets\) pioneered by Silva~\cite{silva2010kleene,silvabonsanguerutten2010nondeterministic}, Myers~\cite{myers2009coalgebraic}, and Milius~\cite{milius2010streamcircuits}.
The processes parametrised framework~\cite{schmidrozowskisilvarot2022processes} is a generalisation of the calculus of Milner that is perpendicular to those offered by Silva, Myers, and Milius.
It aims to capture system types with algebraically characterised branching structures, instead of coalgebras with Kripke polynomial signatures, and only accounts for a simplification of Stark and Smolka's calculus.
The ordered processes parametrised framework is different: it generalises Stark and Smolka's calculus, captures a weakening of Milner's \acro{ARB}, and produces coalgebras and calculi in an entirely different category than \(\Sets\). 
These differences have interesting side-effects, like an intrinsic ordering of process terms that is contained in the similarity order~\cite{hughesjacobs2004simulations} on discrete systems.

We suspect that it is possible to obtain other intrinsic structures on discrete systems by defining processes parametrised-like frameworks in other categories.
One example is metric spaces:
branching fixed-points could be constructed in quantitative algebraic theories~\cite{mardarepanangadenplotkin2021quantitative} using the Banach fixed-point theorem~\cite{bloomesik1997fixedpoints,mardarepanangadenplotkin2021quantitative}.
Any setting in which iterativity is given in abstract algebraic terms should admit a parametrised family of process algebras similar to ours.
We intend to employ the theory of iterative algebras~\cite{bloomE1976iterative} for this purpose in the future.

The idea of relating coinductive properties to branching structure, as in \cref{sec:semantics}, also appears in the metric coalgebra literature, where metric bisimulation is related to behaviour~\cite{vanbreugelworrell2005pseudometric}.
However, the combinatorial study of these relationships with direct reference to the algebraic presentation of the branching monad has seen little attention outside of the recent work of Gumm~\cite{gumm2022freelattice}.
One interesting use case for the algebraic analysis of branching is the characterisation of similarity in \cref{sec:simulations}, which we expect can be generalised to other coinductive properties, such as metric bisimulation or trace equivalence~\cite{hasuoJS2007generic}. 

Lastly, we have several questions about star fragments and polystar fragments.
The most difficult question is whether the axioms in \cref{fig:axioms for stars} are complete for polystar fragments. 
Another question is for which inequational theories do the polystar and star fragments specify the same set of behaviours. 
This is the case for pointed convex algebra and guarded algebra, but is not the case for the theory of guarded convex algebra from \cref{eg:ProbGKAT}.
A study of the different expressiveness levels of fragments of \(\Exp\) would be interesting.

\printbibliography

\appendix


\section{Proofs from \cref{sec:ordered processes}}

\behequivisanequiv*

\begin{proof}
    If \(x_1 =_\beh x_2\), then \(x_1 \le_\beh x_2 \le_\beh x_1\), so we only prove the converse. 

    Conversely, it suffices to consider the case in which \(x_1,x_2 \in X\) for some \((X,\le, \beta)\).
    Let \(h_i : {(X, \le, \beta)} \to (Y_i, \le \vartheta_i)\) for \(i=1,2\) such that \(h_1(x_1) \le h_2(x_2)\) and \(h_2(x_2) \le h_2(x_1)\). 
    Let \(k_i : {(Y_i, \le, \vartheta_i)} \to (W, \le, \varphi)\), \(i=1,2\), be the pushout of \(h_1\) and \(h_2\), and observe that \[
        k_1h_1(x_1) \le k_1h_1(x_2) = k_2h_2(x_2) \le k_2h_2(x_1) = k_1h_1(x_1)
    \]
    It follows that \(k_1h_1\) identifies \(x_1\) with \(x_2\), so indeed \(x_1 =_\beh x_2\). 
\end{proof}

\section{Proofs from \cref{sec:inequational theories}}

\begin{lemma}
	Let \(S^*(V, \le) = (S^*V, \le_S)\) be as it is defined below \cref{rem:peculiar operation order}.
	Then \(\le_S\) is a partial order on \(S^*V\). 
\end{lemma}

\begin{proof}
	By induction on the pair of lengths of the proofs of \(p \le_S q\) and \(q \le_S p\), in the lexicographical order.
	From symmetry, we can assume the shorter of the proofs ends with \(p \le_S q\).
	If the last step in the proof of \(p \le_S q\) is 
	\[
		\prftree[r]{(\(\le_n\))}{\sigma_1 \le_n \sigma_2}{\sigma_1(\vec p) \le_S \sigma_2(\vec p)}
	\]
	then \(p = \sigma_1(\vec p)\) and \(\sigma_2(\vec p)\).
	It follows from \(q \le_s p\) that we also have \(\sigma_2 \le_n \sigma_1\), so in fact \(\sigma_1 = \sigma_2\).
	If the last step in the proof of \(q \le_S p\) is also (\(\le_n\)), then \(\sigma_1 = \sigma_2\), so in fact, \(p = q\).
	If the last step in the proof of \(q \le_S p\) is 
	\[
		\prftree[r]{(C)}{(\forall i)~p_i \le_S q_i}{\sigma(\vec p) \le_S \sigma(\vec q)}	
	\]
	then again, \(\sigma_1 = \sigma = \sigma_2\).
	If the last step is reflexivity, then also \(\sigma_1 = \sigma_2\).
	Lastly, if the last step in the proof of \(q\le_S p\) is
	\[
		\prftree[r]{(T)}{q \le_S r \qquad r \le_S p}{q \le_S p}	
	\]
	then also \(p \le_S r \le_S p\) and \(q \le_S r \le_S q\), so by induction \(p = r = q\). 

	Suppose the last step in the proof of \(p \le_S q\) is (C).
	Then \(p = \sigma(\vec p)\) and \(q = \sigma(\vec q)\).
	By symmetry we have already considered the case in which the last step in the proof of \(q \le_S p\) is (\(\le_n\)).
	If the last step in the proof of \(q \le_S p\) is also (C), then \(p_i = q_i\) for all \(i\) by induction, so indeed \(p = q\).
	If the last step in the proof of \(q \le_S p\) is (T), then again \(p = r = q\) by induction.
	
	Finally, if the last step in the proof of \(p \le_S q\) is (T), then again \(q \le_S p \le_S r \le_S q \le_S p\), so by induction \(p = r = q\).
\end{proof}

\begin{lemma}
	Let \(M\) be the free \((S, \Ie)\)-algebra construction as outlined in \cref{def:free SIalg}.
	Then for any \((S,\Ie)\)-algebra \((Y, \le, \alpha)\) and any \(h : (X, \le) \to (Y, \le)\), there is a unique \(S\)-algebra homomorphism \(h^\alpha : (M(X, \le), \rho) \to (Y, \le, \alpha)\) such that \(h^\alpha \circ \eta = h\).
\end{lemma}

\begin{proof}
	Let \((Y, \le, \alpha)\) be an \((S, \Ie)\)-algebra and \(h : (X, \le) \to (Y, \le)\). 
	Define \(h^\alpha([x]_\Ie) = h(x)\) for each \(x \in X\), and \(h^{\alpha}([\sigma(p_1, \dots, p_n)]_\Ie) = \alpha(\sigma, h^\alpha(p_1), \dots, h^\alpha(p_n))\).
	Then \(h^\alpha\) is clearly an \(S\)-algebra homomorphism satisfying \(h^\alpha \circ \eta = h\).
	If \(k : (M(X, \le), \rho) \to (Y, \le, \alpha)\) is another such homomorphism, one can verify that \(h^\alpha([p]_\Ie) = k([p]_{\Ie})\) by induction on \(p \in S^*(X, \le)\). 
\end{proof}

\lfpishonest*

\begin{proof}
	Define an operator \(\lambda_x p : M(X, \le) \to M(X,\le)\) for each \(p(x, \vec y) \in S^*X\) by setting \(\lambda_xp(q) = p(q, \vec y)\).
	Note that the operator \(\lambda_xp\) is monotone by the \(S\)-term ordering rules.
	We show that \(\lfp_x p\) is the least fixed-point of \(\lambda_x p\). 

	First let us check that \(\lfp_x p\) is a fixed-point. 
	By assumption, \(\lambda_x p(\lfp_xp) = p(\lfp_xp, \vec y) \sqsubseteq_\Ie \lfp_xp\).
	For the converse inequality, observe that \(\lambda_x p(\lfp_xp)\) is also a prefixed-point by monotonicity of \(\lambda_x p\). 
	Since \(\lfp_xp\) is the least such prefixed-point, \(\lfp_xp \sqsubseteq_\Ie \lambda_x p(\lfp_xp)\). 
	It follows that \(\lambda_x p(\lfp_xp) \equiv_{\Ie} \lfp_xp\).
	
	To see that it is the least solution, note that every fixed-point of \(\lambda_x p\) is a prefixed-point of \(\lambda_xp\).
	By assumption, \(\lfp_xp\) is below every other prefixed-point of \(\lambda_xp\). 
\end{proof}

\begin{lemma}\label{lem:lfp monotonicty}
	For any \(x \in X\), the map \(\lfp_x : M(X, \le) \to M(X, \le)\) is monotone. 
\end{lemma}

\begin{proof}
	Let \(p_1, p_2 \in S^*(X, \le)\) and assume \(p_1 \sqsubseteq_\Ie p_2\). 
	We prove that \(\lfp_xp_1 \sqsubseteq_\Ie \lfp_xp_2\) by showing the inclusion \(\Pfp_x(p_2) \subseteq \Pfp_x(p_1)\). 
	Let \(q \in \Pfp_x(p_2)\).
	Then by incongruence,
	\[
	\lambda_xp_1(q) = p_1[q\sub x] \sqsubseteq_\Ie p_2[q\sub x] = \lambda_x p_2(q) \sqsubseteq_\Ie q
	\]		
	It follows that \(\lfp_xp_1 = \min \Pfp_x(p_1) \sqsubseteq_\Ie \min \Pfp_x(p_2) = \lfp_xp_2\).
\end{proof}

\section{On Substitution}

The first kind of substitution that appears is \defn{syntactic} substitution. 
Given two expressions \(e\) and \(f\) and a variable \(v\), we define the expression \(e[f\sub\code v]\) by induction on \(e\) as follows: For the basic constructions, 
\[
	\code u[f\sub\code v] = \begin{cases}
		f &u = v \\
		u &u \neq v
	\end{cases}
	\quad 
	(ae)[f\sub\code v] = a(e[f\sub\code v])
	\quad
	\sigma(e_1,\dots,e_n)[f\sub\code v] = \sigma(e_1[f\sub\code v], \dots, e_n[f\sub\code v])
\]
and 
\[
	(\beta \code u~e)[f\sub \code v] = \begin{cases}
		\beta\code u~e & \code u = \code v \\
		\beta\code u~e[f\sub \code v] & \code u \neq \code v
	\end{cases}
\]
but for the recursion case, we only let \((\mu\code u~e)[f\sub\code v]\) be well-defined if either \(\code u = \code v\), in which case \((\mu \code v~e)[f\sub\code v] = \mu \code v~e\) (because \(\var v\) is not free in \(\mu \code v~e\)), or \(\code u\) is not free in \(f\), in which case \((\mu\code u~e)[f\sub\code v] = \mu\code u~(e[f\sub\code v])\).
Thus, \([f\sub\code v]\) is a partial map \(\Exp \rightharpoonup \Exp\).

We similarly define \(e[f_1\sub\code v_1,\dots,f_n\sub\code v_n]\) for a distinct list of variables \(\code v_1,\dots,\code v_n\) to be the simultaneous substitution of \(f_i\) for \(\code v_i\), \(i=1,\dots,n\).
For this kind of substitution,
\[
\code u[f_1\sub\code v_1,\dots,f_n\sub\code v_n] = \begin{cases}
	f_i &\code u = \code v_i \\
	\code u &(\forall i\le n)~\code u \neq \code v_i
\end{cases}
\]
and if \(u = v_i\), then
\[
(\mu\code u~e)[f_1\sub\code v_1,\dots,f_n\sub\code v_n] = (\mu\code u~e)[f_1\sub\code v_1, \dots, f_{i-1}\sub\code v_{i-1}, f_{i+1}\sub\code v_{i+1},\dots, f_n\sub\code v_n]
\]
and otherwise, if \(u\) is not free in \(f_i\) for all \(i\le n\), then
\[
(\mu\code u~e)[f_1\sub\code v_1,\dots,f_n\sub\code v_n] = \mu\code u~(e[f_1\sub\code v_1,\dots,f_n\sub\code v_n])
\]
Again, \([f_1\sub\code v_1,\dots,f_n\sub\code v_n]\) defines a partial operation \(\Exp \rightharpoonup \Exp\). 

\begin{lemma}
	Let \(e, f \in \Exp\) and \(\code v \in V\).
	If no free variable of \(f\) is bound in \(e\), then \(e[f\sub\code v]\) is well-defined. 
\end{lemma}

\begin{proof}
	See, for example, \cite[Lemma C.8]{full}.
\end{proof}

\begin{lemma}
	The syntactic substitution operators \((-)[-\sub \code v]\) and \((-)[-\gdsub \code v]\) are monotone with respect to \(\le_{\exp}\) in both arguments. 
\end{lemma}

\begin{proof}
	We argue monotonicity of the first syntactic  substitution operator, since monotonicity of guarded syntactic substituton is argued similarly. 
 
	Let \(\preceq\) be any precongruence on \((\Exp, \le_{\exp}, \ev)\). 
	Let \(e \in \Exp\) and \(\code v \in \Var\).
	Monotonicity of the map \(e[-\sub\code v]\) is a direct consequence of the precongruence conditions. 
	Given \(g_1 \preceq g_2\) such that no free variable of either \(g_i\) is bound in \(e\), we show that \(e[g_1\sub\code v] \preceq e[g_1\sub\code v]\) by induction on \(e\). 
	For the base cases, 
	\begin{itemize}
		\item \(\code u[g_1\sub\code v] = \code u = \code u[g_2\sub\code v]\) and similarly for \(0\) in place of \(\code u\), and any expression of the form \(\beta\code v~e\) or \(\mu\code v~e\) in place of \(\code u\), and
		\item \(\code v[g_1\sub \code v] = g_1 \preceq g_2 = \code v[g_2\sub\code v]\).
	\end{itemize}
	For the inductive step,
	\begin{itemize}
		\item \((\code a e)[g_1\sub\code v] = \code a e[g_1\sub\code v] \preceq \code ae[g_2\sub\code v] = (\code ae)[g_2\sub\code v]\),
		\item \(\sigma(\vec e)[g_1\sub\code v] = \sigma((e_i[g_1\sub\code v])_{i \le n}) \preceq \sigma((e_i[g_2\sub\code v])_{i \le n}) = \sigma(\vec e)[g_2\sub\code v]\),
	\end{itemize}
	and similarly for \(\beta\code u~e\) and \(\mu\code u~e\). 

	Monotonicity in the first argument is similarly straightforward but requires discreteness of the poset \((\Var, =)\). 
	Let \(e_1 \le_{\exp} e_2\) and \(g \in \Exp\) such that no free variable of \(g\) appears bound in either \(e_1\) or \(e_2\). 
	The proof proceeds by induction on the deduction of \(e_1 \le_{\exp e_2}\) from the term ordering rules. 
	We only cover the base cases, since the inductive steps follow from the definition of \([g\sub\code v]\) as a (partial) \(\Sigma_M\)-algebra homomorphism. 
	
	If the deduction is 
	\[
		\inferrule{\code a_1 \le_{\act} \code a_2}{\code a_1e \le_{\exp} \code a_2e}
	\] 
	then we find 
	\[
		(\code a_1 e)[g\sub\code v]
		= \code a_1 e[g\sub\code v]	
		\le_{\exp} \code a_2 e[g\sub\code v]	
		= (\code a_2 e)[g\sub\code v]	
	\]
	because the inequality above is also an instance of the rule:
	\[
		\inferrule{\code a_1 \le_{\act} \code a_2}{\code a_1e[g\sub\code v] \le_{\exp} \code a_2e[g\sub\code v]}
	\] 
	Similarly for the inference rule
	\[
		\inferrule{\sigma_1 \le_n \sigma_2}{\sigma_1(\vec e) \le_{\exp} \sigma_2(\vec e)}	\qedhere
	\]
\end{proof}

\begin{remark}
	The issue with having a relation such as \(\code u < \code v\) in \((\Var, \le_{\var})\), say, is that \(\code u = \code u[0\sub \code v]\) while \(\code v[0 \sub \code v] = 0\). 
	That is, even though \(\code u < \code v\), \(\code v[0\sub \code v] < \code u[0 \sub \code v]\), thus \([0\sub\code v]\) fails to be monotone. 
	This is why we assume \((\Var, =)\) is discrete. 
\end{remark}

\paragraph*{Semantic substitution.}
Unlike its syntactic relative, \(\{s\sub\code v\}\) is a total function (and consequently so is \(\{s\gdsub\code v\}\)).
Both are automatically monotone by construction:
formally, semantic substitution is defined as a final coalgebra homomorphism.
Consider the coalgebra 
\[
	[\zeta , \zeta'] : (Z, \le) + (Z\times\{\bullet\}, \le) \longrightarrow B_M((Z, \le) + (Z, \le))	
\]
where on the second component,
\[
	\zeta'(t, \bullet) 
	= \begin{cases}
		\code u &\zeta(t) = \code u \neq \code v \\
		\zeta(s) &\zeta(t) = \code v \\
		\langle \code a, (t', \bullet)\rangle &\zeta(t) = \langle\code a, t'\rangle \\
		\sigma(\zeta'(t_1, \bullet), \dots, \zeta'(t_n, \bullet)) & \zeta(t) = \sigma(\zeta(t_1), \dots, \zeta(t_n))
	\end{cases}
\]
We need to check that \(\zeta'\) is monotone. 
To this end, let \(t \le_\beh t'\). 
Then \(\zeta(t) \sqsubseteq_{\Ie} \zeta(t')\), by monotonicity of \(\zeta\). 
We show monotonicity by induction on the proof of \(\zeta(t) \sqsubseteq_{\Ie} \zeta(t')\). 
For the base cases,
\begin{itemize}
	\item if \(\zeta(t) = \code u\) and \(\zeta(t') = \code u'\), then \(\code u = \code u'\) by discreteness of \(\Var\). 
	In this case, \(\zeta'(t, \bullet) = \zeta'(t', \bullet)\). 
	\item If \(\zeta(t) = \langle a_1, t_0\rangle\) and \(\zeta(t') = \langle a_2, t_0\rangle\) and \(\code a_1 \le_{\act} \code a_2\), then \(\zeta'(t, \bullet) = \langle \code a_1, (t_0, \bullet)\rangle \sqsubseteq_\Ie \langle \code a_2, (t_0, \bullet)\rangle \zeta'(t', \bullet)\).
	\item If \(\zeta(t) = p(\vec{\code v}, \langle\code a_i, t_i\rangle_{i \le n})\) and \(\zeta(t') = q(\vec{\code u}, \langle\code a_i, t_i'\rangle_{i \le n})\) and \((p, q) \in \Ie\), then \[
		\zeta'(t, \bullet) = p(\vec{\code v}, \langle\code a_i, (t_i, \bullet)\rangle_{i \le n})
	\] 
	and \(\zeta'(t', \bullet) = q(\vec{\code u}, \langle\code a_i, (t_i', \bullet)\rangle_{i \le n})\) and therefore \(\zeta'(t,\bullet) \sqsubseteq_\Ie \zeta'(t', \bullet)\) as well. 
\end{itemize}
The inductive step follows from the precongruence of \(\sqsubseteq_\Ie\). 

Thus, we obtain \[
	\beh_{[\zeta, \zeta']} : ((Z, \le) + (Z\times \{\bullet\}, \le), [\zeta, \zeta']) \to (Z, \le_\beh, \zeta)
\]
and define \(\{s \sub \code v\}\) by
\[\begin{tikzcd}[column sep=2cm]
	(Z, \le_\beh) 
	\ar[r, dashed, "\{s\sub\code v\}"]
	\ar[d, "\cong"']
	&
	(Z, \le_\beh)
	\\
	(Z, \le_\beh)\times\{\bullet\}
	\ar[r, "in_r"]
	&
	(Z, \le_\beh) + (Z, \le_\beh)\times\{\bullet\} 
	\ar[u, "\beh_{[\zeta, \zeta']}"']
\end{tikzcd}\]
Guarded syntactic substitution is defined in terms of \(\{s\sub\code v\}\): given \(p(\vec{\code v}, \langle \code a_i, t_i\rangle_{i \le n}) \in {B_M(Z, \le_\beh)}\),
\[
	p(\vec{\code v}, \langle \code a_i, t_i\rangle_{i \le n})\{s \gdsub \code v\} 
	= p(\vec {\code v}, \langle \code a_i, t_i\{s \sub \code v\}\rangle_{i \le n})
\]
Despite their differences, however, behavioural substitution enjoys many of the important properties of syntactic substitution. 
The following recreation of~\cite[Theorem C.6]{schmidrozowskisilvarot2022processes} provides a simplified coinductive principle with which we can prove the desirable properties for our processes.

\begin{theorem}\label{thm:bisimilar maps}
	Let \(h,k : (Z,\le_\beh) \to (Z,\le_\beh)\).
	If \(h\) and \(k\) satisfy (i)-(iii) below, then \(h(t) = k(t)\) for all \(t \in Z\).
	\begin{enumerate}
		\item[(i)] if \(\zeta(t) \in \Var\), then \(h(t) = k(t)\);
		\item[(ii)] if \(\zeta(t) = \langle \code a, r\rangle \in \Act\times Z\), then \(\zeta(h(t)) = \langle\code a,h(r)\rangle\) and \(\zeta(k(t)) = \langle \code a,k(r)\rangle\); and
		\item[(iii)] if \(\zeta(t) = \sigma(\zeta(t_1), \dots, \zeta(t_n))\), then 
		\[
		\zeta(h(t)) = \sigma(\zeta(h(t_1)), \dots, \zeta(h(t_n)))
		\quad
		\text{and}
		\quad
		\zeta(k(t)) = \sigma(\zeta(k(t_1)), \dots, \zeta(k(t_n)))
		\] 
	\end{enumerate}
\end{theorem}

\begin{proof}
	We proceed with a proof by \emph{coinduction}. 
	Namely, we will give the relation\[
	R := \Delta_Z \cup \{(h(r), k(r)) \mid r \in Z\}
	\]
	a \(B_M\)-coalgebra structure \(\rho : (R, =) \to B_M(R, =)\) such that the projections \(\pi_i : (R, =) \to (Z, \le_\beh)\), \(i = 1,2\), are coalgebra homomorphisms. 
	By finality of \((Z,\le_\beh,\zeta)\), this implies that \(\pi_1 = {!}_{\rho} = \pi_2\), which establishes the identity we are hoping to prove.
	
	Let \(t \in Z\) and \(
	\zeta(t) = p(\code v_1,\dots,\code v_m, \langle\code a_i,s_i\rangle_{i \le n})
	\)
	for some \(p \in B_M(Z, \le_\beh)\).
	Along the diagonal, define
	\[
	\rho(t,t) = p(\code v_1,\dots,\code v_m, \langle\code a_i,(s_i, s_i)\rangle_{i \le n})
	\]
	For the other pairs, assume \(\code v_i = \zeta(h(t_i)) = \zeta(k(t_i))\) for \(i = 1,\dots, m\) using (i) and write
	\[
		\rho(h(t), k(t)) = p(\code v_1, \dots, \code v_m, \langle\code a_i, (h(s_i), k(s_i))\rangle_{i \le n})
	\]   
	using (ii) and (iii).

	To see that \(\pi_i\) is a coalgebra homomorphism for each \(i = 1,2\), observe
	\begin{align*}
		B_M(\pi_1)(\rho(h(t),k(t)))
		&= B_M(\pi_1)(p(\code v_1, \dots, \code v_n, \langle\code a_i, (h(s_i), k(s_i))\rangle_{i \le n})) \\
		&= p(\code v_1, \dots, \code v_n, \langle\code a_i, h(s_i)\rangle_{i \le n}) \\
		&= \zeta \circ \pi_1(h(r), k(r))
	\end{align*}
	and similarly for \(\pi_2\).
\end{proof} 

\begin{definition}
	A return variable \(\code u\) is said to be \defn{dead} in \(t\) if for any \(s \in Z\), \(t\{s\sub \code u\} = t\).
\end{definition}

\begin{lemma}
	If \(\code v\) is dead in \(t\), then \(\lfp_{\code v} \zeta(t) = \zeta(t)\).
\end{lemma}

\begin{proof}
	If \(\code v\) is dead in \(t\), then a representative of \(\zeta(t)\) can be chosen of the form
	\[
	p(\vec{\code u}, \langle\code a_i, t_i\rangle_{i\le n})
	\]
	with \(\code v \notin \vec{\code u}\) and such that \(\code v\) is dead in each \(t_i\) (this can be shown by induction on \(\zeta(t)\)). 
	Thus, if \(\code v\) is dead in \(t\), then
	\[
	\lfp_{\code v} \zeta(t) 
	= p(\vec{\code u}, \langle\code a_i, t_i\{s\sub\code v\}\rangle_{i\le n})
	= p(\vec{\code u}, \langle\code a_i, t_i\rangle_{i\le n}) = \zeta(t)
	\qedhere
	\]
\end{proof}

\begin{lemma}
	If \(\code v\) is dead in \(s\), then \(t\{r \sub \code v\}\{s \sub\code u\} = t\{s\sub\code u\}\{r\{s \sub\code u\}\sub\code v\}\).
\end{lemma}

\begin{proof}
	See \cite[Lemma C.3]{schmidrozowskisilvarot2022processes}.
\end{proof}

As an immediate corollary to the above lemma, we obtain that for any \(p \in B_M(Z, \le_{\beh})\), \[
p\{r \gdsub \code v\}\{s\sub \code u\} = p\{s\sub \code u\}\{r\{s\sub \code u\} \gdsub \code v\}
\] 
when \(\code v\) is dead in \(s\).
This can be shown by induction on \(p\).

\begin{lemma}\label{lem: gamma and substitution}
	Let \(t,s \in Z\), \(\code u, \code v \in \Var\). 
	If \(\code u\) is dead in \(s\) and \(\code u \neq \code v\), then
	\begin{enumerate}
		\item[(i)] \(\gamma(\beta\code u, t\{s\sub\code v\}) = \gamma(\beta\code u, t)\{s \sub\code v\}\)
		\item[(ii)] \(\gamma(\mu\code u, t\{s\sub\code v\}) = \gamma(\mu\code u, t)\{s\sub\code v\}\)
	\end{enumerate}
\end{lemma}

\begin{proof}
	Let \(\zeta(t) = p(\code u, \code v, \vec{\code w}, \langle\code a_i, t_i\rangle_{i \le n})\) and \(\zeta(s) = q(\code v, \vec{\code w}, \langle\code b_i, s_i\rangle_{i \le m})\).
	For (i), define 
	\[
		r = r(\code u, \code v, \vec{\code w}, \langle\code b_i, s_i\rangle_{i \le m}, \langle\code a_i, t_i\{s\sub\code v\}\rangle_{i \le n})
		= p(\code u, q(\code v, \vec{\code w}, \langle\code b_i, s_i\rangle_{i \le m}), \vec{\code w}, \langle\code a_i, t_i\{s\sub\code v\}\rangle_{i \le n}))
	\]
	and observe that \(r = \zeta(t\{s\sub\code v\})\) and \(\lfp_{\code u} r = \lfp_{\code u} p(\zeta(s), \vec{\code w}, \langle\code a_i, t_i\rangle_{i \le n})\).
	Therefore,
	\begin{align*}
		\zeta(\gamma(\beta\code u, t\{s\sub\code v\}))
		&= \lfp_{\code u} \zeta(t\{s\sub\code v\}) \\
		&= \lfp_{\code u} r(\code u, \code v, \vec{\code w}, \langle\code b_i, s_i\rangle_{i \le m}, \langle\code a_i, t_i\{s\sub\code v\}\rangle_{i \le n}) \\
		&= \lfp_{\code u} p(\zeta(s), \langle\code a_i, t_i\{s\sub\code v\}\rangle_{i \le n}) \\
		&= \zeta(\gamma(\beta\code u, t)\{s\sub\code v\})
	\end{align*}
	Property (ii) readily follows. 
	Indeed, 
	\begin{align*}
		\zeta(\gamma(\mu\code u, t\{s\sub\code v\}))
		&= \lfp_{\code u} \zeta(t\{s\sub\code v\})\{\gamma(\mu\code u, t\{s \sub\code v\})\gdsub\code u\} \\
		&= \zeta(\gamma(\beta\code u, t\{s\sub\code v\}))\{\gamma(\mu\code u, t\{s \sub\code v\})\gdsub\code u\} \\
		&= \zeta(\gamma(\beta\code u, t)\{s\sub\code v\})\{\gamma(\mu\code u, t\{s \sub\code v\})\gdsub\code u\} \\
		&= \zeta(\gamma(\beta\code u, t))\{s\sub\code v\}\{\gamma(\mu\code u, t\{s \sub\code v\})\gdsub\code u\} \\
		&= \lfp_{\code u}\zeta(t)\{s\sub\code v\}\{\gamma(\mu\code u, t\{s \sub\code v\})\gdsub\code u\} \\
		&= \lfp_{\code u}\zeta(t)\{\gamma(\mu\code u, t)\gdsub\code u\}\{s \sub\code v\} \tag{\(\code u\) dead in \(s\)}\\
		&= \zeta(\gamma(\mu\code u, t))\{s \sub\code v\}\qedhere
	\end{align*}
\end{proof}

Semantic and syntactic substitution interact as follows. 

\begin{lemma}
	For any \(e \in \Exp\), \(\code v \in \Var\), and \(g \in \Exp\) such that \(\bv(e) \cap \fv(g) = \emptyset\). 
	Then \(\sem{e[g\sub \code v]} = \sem e\{\sem g\sub\code v\}\). 
\end{lemma}

\begin{proof}
	By induction on \(e\). 
	\begin{itemize}
		\item For the base case, let \(\code u \in \Var\). 
		If \(\code v = \code u\), then 
		\[
			\zeta(\sem{\code v[g\sub \code v]})
			= \zeta(\sem{g})
			= \zeta(\code v\{\sem{g}\sub \code v\})
		\]
		as desired. 
		If \(\code v \neq \code u\), then 
		\[
			\zeta(\sem{\code u[g\sub \code v]})
			= \zeta(\sem{\code u})
			= \code u
			= \zeta(\sem{\code u}\{\sem{g}\sub \code v\})
		\]
		as well. 
		
		\item For the first inductive step, let \(\code a \in \Act\) and assume the result for \(e\). 
		\begin{align*}
			\zeta(\sem{\code a e[g \sub \code v]})
			&= \langle \code a, \sem{e[g \sub \code v]}\rangle \\
			&= \langle \code a, \sem{e}\{\sem g \sub \code v\}\rangle \\
			&= \zeta(\gamma(\code a, \sem{e}\{\sem g \sub \code v\})) \\
			&= \zeta(\sem{\code a e}\{\sem g \sub \code v\}) 
		\end{align*}
	
		\item For the second inductive step, let \(\sigma \in I_n\) and assume the result for \(e_1, \dots, e_n\). 
		\begin{align*}
			\zeta(\sem{\sigma(e_1, \dots, e_n)[g \sub \code v]})
			&= \zeta(\gamma (\sigma, \sem{e_1[g \sub \code v]}, \dots, \sem {e_n[g \sub \code v]})) \\
			&= \zeta(\gamma (\sigma, \sem{e_1}\{\sem g \sub \code v\}, \dots, \sem {e_n}\{\sem g \sub \code v\})) \\
			&= \zeta(\gamma (\sigma, \sem{e_1}, \dots, \sem {e_n})\{\sem g \sub \code v\}) \\
			&= \zeta(\sigma(\sem{e_1}, \dots, \sem {e_n})\{\sem g \sub \code v\}) \\
			&= \zeta(\sem{\sigma(e_1, \dots, e_n)}\{\sem g \sub \code v\}) 
		\end{align*}
	
		\item For the third step, let \(\code u \in \Var\) and consider \(\beta\code u~e[g\sub \code v]\). 
		Since \(\bv(e) \cap \fv(g) = \emptyset\), we must have \(\code u \notin \fv(g)\).
		This means that \(\code u\) is dead in \(\sem g\), so
		\begin{align*}
			\zeta(\sem{\beta \code u~e[g \sub \code v]})
			&= \zeta(\gamma(\beta \code u, \sem{e[g \sub \code v]})) \\
			&= \lfp_{\code u} \zeta(\sem{e[g \sub \code v]}) \\
			&= \lfp_{\code u} \zeta(\sem{e}\{\sem g \sub \code v\}) \\ 
			&= \lfp_{\code u} \zeta(\sem{e}\{\sem g \sub \code v\}) \\
			&= (\lfp_{\code u} \zeta(\sem{e}))\{\sem g\sub \code v\} \\
			&= \zeta(\gamma(\beta{\code u}, \sem{e}))\{\sem g\sub \code v\} \\
			&= \zeta(\sem{\beta{\code u}~e})\{\sem g\sub \code v\} \\
			&= \zeta(\sem{\beta{\code u}~e}\{\sem g\sub \code v\}) 
		\end{align*}
		
		\item For the last step, let \(\code u \in \Var\) and consider \(\mu\code u~e[g\sub \code v]\). 
		We have
		\begin{align*}
			\zeta(\sem{\mu \code u~e[g \sub \code v]})
			&= \zeta(\gamma(\mu \code u, \sem{e[g \sub \code v]})) \\
			&= \lfp_{\code u} \zeta(\sem{e[g \sub \code v]})\{\gamma(\mu \code u, \sem{e[g\sub\code v]})\gdsub\code u\} \\		
			&= \lfp_{\code u} \zeta(\sem{e[g \sub \code v]})\{\sem{\mu \code u~e[g\sub \code v]}\gdsub\code u\} \\		
			&= \lfp_{\code u} \zeta(\sem{e}\{\sem g\sub \code v\})\{\sem{\mu \code u~e[g\sub \code v]}\gdsub\code u\}	
		\end{align*}
		Whence, \(\sem{\mu \code u~e[g \sub \code v]}\) is a solution to the behavioural differential equation
		\[
			\zeta(t) = \lfp_{\code u} \zeta(\sem{e}\{\sem g\sub \code v\})\{t\gdsub\code u\}
		\] 
		By definition, then, 
		\[
			\sem{\mu\code u~e[g\sub\code v]} = \gamma(\mu\code u, \sem{e}\{\sem g\sub \code v\})
		\]
		Now, since \(\bv(e) \cap \fv(g) = \emptyset\), again we must have \(\code u \notin \fv(g)\).
		This means that \(\code u\) is dead in \(\sem g\), so finally
		\[
		\gamma(\mu\code u, \sem{e}\{\sem g\sub \code v\}) = \gamma(\mu\code u, \sem{e})\{\sem g \sub \code v\} = \sem{\mu\code u~e}\{\sem g\sub \code v\}\qedhere
		\]
	\end{itemize}
\end{proof}

\begin{lemma}
	For any \(p \in B_M\Exp\), \(g \in \Exp\), \(\code v \in \Var\), such that \(p[g \gdsub \code v]\) is well-defined and \(p = \epsilon(e)\), \[
		B_M(\sem-) \circ \epsilon(e[g\gdsub \code v]) = B_M(\sem-)(p)\{\sem g \gdsub \code v\}
	\]
\end{lemma}

\begin{proof}
	By induction on \(p\). 
	\begin{itemize}
		\item Let \(\code u \in \Var\).
		Then if \(p = \code u\),
		\begin{align*}
			B_M(\sem-) (p[g\gdsub \code v])
			&= B_M(\sem-)(\code u[g\gdsub \code v]) \\
			&= B_M(\sem-)(\code u) \\
			&= \code u \\
			&= \code u\{\sem g \gdsub \code v\} \\
			&= B_M(\sem-) (p)\{\sem g \gdsub \code v\}
		\end{align*}
	
		\item Let \(\code a \in \Act\) and \(e \in \Exp\) such that \(\fv(g)\cap\gd(e) = \emptyset\). 
		\begin{align*}
			B_M(\sem-)(p[g\gdsub \code v])
			&= B_M(\sem-)(\langle\code a, e\rangle[g\gdsub \code v]) \\
			&= B_M(\sem-)(\langle\code a, e[g\sub \code v]\rangle) \\
			&= \langle\code a, \sem{e[g\sub \code v]}\rangle \\
			&= \langle\code a, \sem{e}\{\sem{g}\sub \code v\}\rangle \\
			&= \langle\code a, \sem{e}\rangle \{\sem{g}\gdsub \code v\}\\
			&= B_M(\sem-)(p) \{\sem{g}\gdsub \code v\}
		\end{align*}
	
		\item Let \(\sigma \in I_n\) and assume the result for \(p_1, \dots, p_n\). 
		\begin{align*}
			B_M(\sem-)(\sigma(p_1, \dots, p_n)[g\gdsub \code v])
			&= B_M(\sem-)(\sigma(p_1[g\gdsub \code v], \dots, p_n[g\gdsub \code v])) \\
			&= \sigma(B_M(\sem-)(p_1[g\gdsub \code v]), \dots, B_M(\sem-)(p_n[g\gdsub \code v])) \\
			&= \sigma(B_M(\sem-)(p_1)\{g\gdsub \code v\}, \dots, B_M(\sem-)(p_n)\{g\gdsub \code v\}) \\
			&= \sigma(B_M(\sem-)(p_1), \dots, B_M(\sem-)(p_n))\{g\gdsub \code v\} \\
			&= B_M(\sem-)(\sigma(p_1, \dots, p_n))\{g\gdsub \code v\}\qedhere
		\end{align*}
	\end{itemize}
\end{proof}

\section{Proofs from \cref{sec:semantics}}

In order to guarantee the existence of the denotational semantics \(\sem- : (\Exp, \le_{\exp}) \to (Z, \le)\), we need to establish the following.

\begin{restatable}{lemma}{behaviouralgebramonotone}
	The operation map \(\gamma : \Sigma_SZ \to Z\) is monotone. 
\end{restatable}

\begin{proof}
	It suffices to see that \(\zeta \circ \gamma\) is monotone. 

	The action and operator orders are easly seen to be preserved. 
	For congruence, let \(t_i \le_\beh s_i\) for \(i \le n\) and \(\sigma \in I_n\). 
	Then 
	\begin{align*}
		\zeta \circ \gamma(\sigma, t_1, \dots, t_n)
		= \sigma(\zeta(t_1), \dots, \zeta(t_n)) 
		\sqsubseteq_\Ie \sigma(\zeta(s_1), \dots, \zeta(s_n)) 
		= \zeta \circ \gamma(\sigma, s_1, \dots, s_n)
	\end{align*}
	by monotonicity of \(\zeta\). 
	
	For the fixed-point operators, let \(t_1 \le_\beh t_2\) and \(\code v \in \Var\). 
	\begin{align*}
		\zeta \circ \gamma(\beta\code v, t_1)
		= \lfp_{\code v} \zeta(t_1)
		\sqsubseteq_\Ie \lfp_{\code v} \zeta(t_2)
		= \zeta\circ \gamma(\beta\code v, t_2)
	\end{align*}
	Monotonicity of \(\gamma(\mu\code v, -)\) is trickier: it follows from the more general observation that if \(t_1 \le_\beh t_2\), \(\code v\) is guarded in both \(t_i\) (ie. there is a representative of \(\zeta(t_i)\) of the form \(p(\vec{\code u}, \langle\code a_j, r_j\rangle_{j\le n})\) with \(\code v \notin \vec{\code u}\)), and \(\zeta(s_i) = \zeta(t_i)\{s_i \gdsub \code v\}\) for \(i = 1,2\), then \(s_1 \le_\beh s_2\). 
	
	To show this, we need to reach into the \emph{final sequence} of \(B_M\), the \((\omega+\omega)^{op}\)-chain in the top row of the following diagram
	\[\begin{tikzcd}[row sep=huge]
		Z_0 & Z_1 \ar[l, "\alpha_{10}"'] & Z_2 \ar[l,"\alpha_{21}"'] & \cdots \ar[l] & Z_\omega \ar[l, "\alpha_{\omega n}"'] & \cdots \ar[l] & Z \ar[l, "\alpha_{n}"']\\
		Z \ar[r, "\zeta"'] \ar[u, "\beh_\zeta^0"] \ar[ur, "\beh_\zeta^1"] \ar[urr, "\beh_\zeta^2"] \ar[urrrr, "\beh_\zeta^\omega"]
		& B_MZ   \ar[r, "\zeta^2"']
		& B_M^2Z \ar[r]
		& \cdots \ar[r]
		& B_M^\omega Z \ar[u] \ar[r, "\zeta^{\omega+1}"']
		& \cdots 
	\end{tikzcd}\]  
	in which \(Z_0 = (1, =)\) and \(\beh_\zeta^0\) is the unique map \((Z, \le_\beh) \to (1, =)\), \(Z_{n+1} = B_M(Z_n, \le)\) and \(\beh_\zeta^{n+1} = B_M(\beh_\zeta^n)\circ \zeta\) , everything in the top row is the limit of the cone to its left, and everything in the bottom row is the colimit of the cocone to its left. 
	We refer to \(\beh_\zeta^n(t)\) as the \emph{\(n\)th truncation} of \(t\).
	It was shown originally by Worrell that \((Z, \le_\beh)\) is the limit of the diagram in the top row~\cite{worrell1999terminal}.

	By the explicit construction of limits in \(\Pos\), it suffices to show that \(\beh_\zeta^n(s_1) \le \beh_\zeta^n(s_2)\) for all \(n < \omega+\omega\).  
	Now, \(\beh_\zeta^0(s_1) \le \beh_\zeta^0(s_2)\) trivially, so we can immediately move on to the inductive steps.
	If \(\beh_\zeta^n(s_1) \le \beh_\zeta^n(s_2)\), then 
	\begin{align*}
		\beh_\zeta^{n+1}(s_1) 
		&= B_M(\beh_\zeta^n) \circ \zeta(s_1) \\
		&= B_M(\beh_\zeta^n)(\zeta(t_1)\{s_1 \gdsub \code v\}) \\
		&= B_M(\beh_\zeta^n)(\zeta(t_1))\{\beh_\zeta^n(s_1) \gdsub \code v\} \tag{\(\code v\) gdd in \(t_i\)}\\
		&\sqsubseteq_{\Ie} B_M(\beh_\zeta^n)(\zeta(t_2))\{\beh_\zeta^n(s_1) \gdsub \code v\} \\
		&\sqsubseteq_{\Ie} B_M(\beh_\zeta^n)(\zeta(t_2))\{\beh_\zeta^n(s_2) \gdsub \code v\} \tag{ih}\\
		&= \beh_\zeta^{n+1}(s_2)
	\end{align*}
	Note that we use a minor abuse of notation above to denote the substitution of the \(n\)th truncation of \(s_i\) for \(\code v\)  in the second component of the transitions \(\langle \code a, r\rangle\) that appear in the (truncated) branching of \(t_i\) (or possibly \(m\)th truncation \(s_i\) for \(m < n\) at further truncated sections of \(t_i\)).
	This is defined analogously to guarded semantic substitution, so we write \(B_M(\beh_\zeta^n)(\zeta(t_1))\{\beh_\zeta^n(s_1) \gdsub \code v\}\) to denote this substitution. 

	At the \(\omega^{op}\)-limit, \(\beh^\omega_\zeta(s_1) \le \beh^\omega_\zeta(s_2)\) because \(\beh^n_\zeta(s_1)\le \beh^n_\zeta(s_1)\) for all \(n < \omega\).
	This completes the proof because \(s_1 = \beh_\zeta(s_1) \le_\beh \beh_\zeta(s_2) = s_2\) if and only if \(\beh_\zeta^n(s_1) \le_\beh \beh_\zeta^n(s_2)\) for all \(n < \omega + \omega\). 
	So, in particular, 
	\begin{align*}
		\zeta\circ \gamma(\mu\code v, t_1) 
		&= \lfp_{\code v} \zeta(t_1)\{\gamma(\mu\code v, t_1)\gdsub\code v\}  \\
		&= \zeta(\gamma(\beta\code v, t_1))\{\gamma(\mu\code v, t_1)\gdsub\code v\} \\
		&\sqsubseteq_\Ie \zeta(\gamma(\beta\code v, t_2))\{\gamma(\mu\code v, t_2)\gdsub\code v\} \\
		&= \zeta\circ \gamma(\mu\code v, t_2) 
	\end{align*}
	if \(t_1 \le_\beh t_2\).
\end{proof}

\adequacytheorem*

\begin{proof}
	The proof is essentially the same as that of \cite[Theorem 1]{schmidrozowskisilvarot2022processes}, with the added case of the unguarded loop. 
	That is, we need to show that \(B_M(\sem-)\circ \epsilon(\beta\code v~e) = \zeta \circ \gamma(\beta\code v, \sem{e})\) given the analogous equation for \(e\). 
	\begin{align*}
		\zeta(\gamma(\beta\code v, \sem{e}))
		&= \lfp_{\code v} \zeta(\sem e) \\
		&= \lfp_{\code v} B_M(\sem-)(\epsilon(e)) \\
		&= B_M(\sem-)(\lfp_{\code v} \epsilon(e)) \\
		&= B_M(\sem-)(\epsilon(\beta\code v~e)) 
	\end{align*}
	by condition \textsf{2} of \cref{def:iterative}. 
\end{proof}

\section{Proofs from \cref{sec:a family of ordered process calculi}}

\begin{lemma}\label{lem:epsilon is monotone}
	The small-step semantics map \(\epsilon\) is monotone. 
\end{lemma}

\begin{proof}
	We show that \(e \sqsubseteq_{\exp} f\) implies \(\epsilon(e) \sqsubseteq_\Ie \epsilon(f)\) by induction on the proof of \(e \sqsubseteq_{\exp} f\).
	The only interesting cases are in the inductive step concerning congruence of \(\beta\code v\) and \(\mu\code v\), since \(\epsilon\) (by definition) commutes with the \(S\)-operations.
	  
	By induction and the monotonicity of \(\lfp_{\code v}\), we have
	\[
	\epsilon(\beta\code v~e)
	= \lfp_{\code v} \epsilon(e) 
	\sqsubseteq_\Ie \lfp_{\code v} \epsilon(f)
	= \epsilon(\beta\code v~f)
	\]
	From monotonicity of \((-)[-\gdsub\code v]\), we glean that
	\[
	\epsilon(\mu\code v~e)
	= \lfp_{\code v} \epsilon(e)[\mu\code v~e\gdsub \code v] 
	\sqsubseteq_\Ie \lfp_{\code v} \epsilon(f)[\mu\code v~e\gdsub \code v] 
	\sqsubseteq_\Ie \lfp_{\code v} \epsilon(f)[\mu\code v~f\gdsub \code v] 
	= \epsilon(\mu\code v~f)\qedhere
	\] 
\end{proof}

\section{Proofs from \cref{sec:axiomatisation}}

\intermediatesoundnesstheorem*

\begin{proof}
	Uniqueness follows from the surjectivity of \([-]_\equiv\), so it suffices to show existence.
	We use a diagonal fill-in to produce the map \(\bar\epsilon\) in the diagram
	\[\begin{tikzcd}
		(\Exp, \sqsubseteq_{exp}) \ar[d, "\epsilon"'] \ar[r, "{[-]_{\equiv}}"] & (\E, \sqsubseteq) \ar[d, dashed, "\bar\epsilon"] \\
		B_M(\Exp, \sqsubseteq_{exp}) \ar[r] & B_M(\E, \sqsubseteq)
	\end{tikzcd}\]
	This requires us to show that whenever \(e \sqsubseteq f\), we must also have the inequality \(B_M([-]_\equiv)\circ \epsilon(e) \sqsubseteq_\Ie B_M([-]_\equiv)\circ \epsilon(f)\).
	We show this by induction on the proof of \(e \sqsubseteq f\). 
	\begin{description}
		\item[(\acro{\Ie})] Suppose \(e \sqsubseteq f\) follows from \(e \sqsubseteq_\Ie f\). Then, since \(\epsilon\) is an \(S\)-algebra homomorphism, \(\epsilon(e) \sqsubseteq_\Ie \epsilon(f)\). 
		By monotonicity of \(B_M([-]_\equiv)\), \(B_M([-]_\equiv)\circ \epsilon(e) \sqsubseteq_\Ie B_M([-]_\equiv)\circ \epsilon(f)\).
		
		\item[(\acro{S})] Suppose \(e = \sigma_1(e_1, \dots, e_n)\) and \(f = \sigma_2(e_1, \dots, e_n)\) and \(\sigma_1 \le_S \sigma_2\). 
		Again, \(\epsilon\) is an \(S\)-algebra homomorphism, so is by definition monotonic wrt to the operation order.
		It follows that \(\epsilon(e) \sqsubseteq_\Ie \epsilon(f)\).
		By monotonicity of \(B_M([-]_\equiv)\), \(B_M([-]_\equiv)\circ \epsilon(e) \sqsubseteq_\Ie B_M([-]_\equiv)\circ \epsilon(f)\).
		
		\item[(\acro{Act})] Suppose \(e = \code a_1 e_1\) and \(f = \code a_2 e_1\).
		Then by monotonicity of \(\langle-\rangle\), \(\langle \code a_1, e_1\rangle \le \langle \code a_2, e_1\rangle\). 
		Again, it follows that \(\epsilon(e) \sqsubseteq_\Ie \epsilon(f)\). 
		
		\item[(\acro{R1a})] Suppose that \(\code v \in \gd(g_i)\) for \(i \le n\) and \(e = \beta\code v~p(\code v, \vec g)\) and \(f = (\lfp_{\code v} p)(\vec g)\). 
		Then 
		\[
		\epsilon(e) 
		= (\lfp_{\code v} p)(\epsilon(g_1), \dots, \epsilon(g_n))
		= \epsilon((\lfp_{\code v} p)(g_1, \dots, g_n))
		\]
		because \(\epsilon\) is an \(S\)-algebra homomorphism.
		Therefore, \(B_M([-]_\equiv)\circ \epsilon(e) = B_M([-]_\equiv)\circ \epsilon(f)\) as desired.
		
		\item[(\acro{R1b})] Suppose \(e = \mu \code v~e_1\) and \(f = \beta \code v~e_1[e\gdsub \code v]\).
		Then \[
		\epsilon(e) 
		= \lfp_{\code v} \epsilon(e_1)[e\gdsub \code v]
		= \epsilon(\beta\code v~e_1)[e\gdsub \code v]
		= \epsilon(\beta\code v~e_1[e\gdsub \code v])
		\]
		because \(\epsilon\) is an \(S\)-algebra homomorphism.
		Therefore, \(B_M([-]_\equiv)\circ \epsilon(e) = B_M([-]_\equiv)\circ \epsilon(f)\) as desired.
	\end{description}
	This concludes the base case.
	For the inductive case, assume  \(B_M([-]_\equiv)\circ \epsilon(e_i) \le B_M([-]_\equiv)\circ \epsilon(f_i)\) for all \(i \le n\). 
	We skip the (\acro{P}i) cases, as these are trivially true by assumption. 
	\begin{description}
		\item[(\acro{C1})] Suppose \(e = \sigma(e_1, \dots, e_n)\) and \(f = \sigma(f_1, \dots, f_n)\). 
		Since \(\epsilon\) is an \(S\)-algebra homomorphism and \(B([-]_\equiv)\) is a monotone \(S\)-algebra homomorphism, 
		\begin{align*}
			B([-]_\equiv) \circ \epsilon(e) 
			&= B([-]_\equiv) (\sigma(\epsilon(e_1), \dots, \epsilon(e_n))) \\
			&= \sigma (B([-]_\equiv) \circ \epsilon(e_1), \dots, B([-]_\equiv) \circ \epsilon(e_n)) \\
			&\sqsubseteq_\Ie \sigma (B([-]_\equiv) \circ \epsilon(f_1), \dots, B([-]_\equiv) \circ \epsilon(f_n)) \\
			&= B([-]_\equiv) \circ \epsilon(f)
		\end{align*}
		
		\item[(\acro{C2})] Suppose \(e = \code a e_1\) and \(f = \code a f_1\). 
		Then 
		\begin{align*}
			B([-]_\equiv) \circ \epsilon(e) 
			&= B([-]_\equiv) (\langle \code a, e \rangle) \\
			&= \langle \code a, [e]_\equiv \rangle \\
			&\le \langle \code a, [f]_\equiv \rangle \\
			&= B([-]_\equiv) \circ \epsilon(f)
		\end{align*}
		by monotonicity of \(\eta\).
		
		\item[(\acro{C3})] Suppose \(e = \beta\code v~e_1\) and \(f = \beta\code v~f_1\). 
		We have 
		\begin{align*}
			B([-]_\equiv) \circ \epsilon(e) 
			&= B([-]_\equiv) (\lfp_{\code v}\epsilon(e_1)) \\
			&= \lfp_{\code v} B_M([-]_{\equiv}) \circ \epsilon(e_1)\\ 
			&\sqsubseteq_\Ie \lfp_{\code v} B_M([-]_{\equiv}) \circ \epsilon(f_1)\\
			&= B([-]_\equiv) \circ \epsilon(f)
		\end{align*}
		
		\item[(\acro{C4})] Suppose \(e = \mu\code v~e_1\) and \(f = \mu\code v~f_1\). 
		Then 
		\begin{align*}
			B([-]_\equiv) \circ \epsilon(e) 
			&= B([-]_\equiv) (\lfp_{\code v}\epsilon(e_1)[e\gdsub \code v]) \\
			&= B([-]_\equiv)(\lfp_{\code v}\epsilon(e_1))[[e]_\equiv\gdsub \code v] \\
			&= \lfp_{\code v} B([-]_\equiv)\circ \epsilon(e_1)[[e]_\equiv\gdsub \code v] \\
			&\sqsubseteq_\Ie \lfp_{\code v} B([-]_\equiv)\circ \epsilon(f_1)[[e]_\equiv\gdsub \code v] \\
			&\sqsubseteq_\Ie \lfp_{\code v} B([-]_\equiv)\circ \epsilon(f_1)[[f]_\equiv\gdsub \code v] \\
			&= B([-]_\equiv) \circ \epsilon(f)
		\end{align*}
		
		\item[(\acro{R2a})] Suppose \(e = \beta\code v~p(v, \vec f)\) and \(f = g\) and \(p(g, \vec f) \sqsubseteq g\), and assume that
		\[
		B_M([-]_\equiv)\circ \epsilon(p(g, \vec f)) \sqsubseteq_\Ie B_M([-]_\equiv)\circ \epsilon(g).
		\] 
		Then we have 
		\begin{align*}
			B([-]_\equiv) \circ \epsilon(e) 
			&= B([-]_\equiv) (\lfp_{\code v}\epsilon(p(\code v, \vec f))) \\
			&= \lfp_{\code v} (B([-]_\equiv) \circ \epsilon(p(\code v, \vec f))) \\
			&= \lfp_{\code v} (B([-]_\equiv) \circ p(\code v, \epsilon(\vec f))) \\
			&= \lfp_{\code v} (p(v, B([-]_\equiv) \circ \epsilon(\vec f))) \\
			&= \lfp_{\code v} p(B([-]_\equiv) \circ \epsilon(\vec f)) \\
		\end{align*}
		In other words, \(B([-]_\equiv) \circ \epsilon(e)\) is the least fixed-point of \(p(\code v, B([-]_\equiv) \circ \epsilon(\vec f))\), so it suffices to show that \(B_M([-]_\equiv)\circ \epsilon(g)\) is another prefixed-point. 
		To this end, compute
		\begin{align*}
			B_M([-]_\equiv)\circ \epsilon(p(g, \vec f))
			&= B_M([-]_\equiv) (p(\epsilon(g), \epsilon(\bar f))) \\
			&= p(B_M([-]_\equiv)\circ \epsilon(g), B_M([-]_\equiv)\circ\epsilon(\bar f)) 
		\end{align*}
		By assumption, then,
		\[
		p(B_M([-]_\equiv)\circ \epsilon(g), B_M([-]_\equiv)\circ\epsilon(\bar f)) \sqsubseteq_\Ie B_M([-]_\equiv)\circ \epsilon(g)
		\]
		It follows that \(B([-]_\equiv) \circ \epsilon(e) \sqsubseteq_\Ie B([-]_\equiv) \circ \epsilon(g)\).
		
		\item[(\acro{R2b})] Suppose \(e = \mu\code v~e_1\) and \(f = g\) and \(\beta\code v~e_1[g\gdsub\code v] \sqsubseteq g\), and assume that
		\[
		B_M([-]_\equiv)\circ \epsilon(\beta\code v~e_1[g\sub \code v]) \sqsubseteq_\Ie B_M([-]_\equiv)\circ \epsilon(g).
		\] 
		We have 
		\begin{align*}
			B_M([-]_\equiv) \circ \epsilon(e)
			&= B_M([-]_\equiv)(\lfp_{\code v} \epsilon(e_1)[e\gdsub\code v]) \\
			&= B_M([-]_\equiv)(\beta{\code v}~ \epsilon(e_1)) [[e]_\equiv\gdsub\code v]\\
			&\sqsubseteq_\Ie B_M([-]_\equiv)(\beta{\code v}~ \epsilon(e_1)) [[g]_\equiv\gdsub\code v]\\
			&= B_M([-]_\equiv)(\beta{\code v}~ \epsilon(e_1)[g\gdsub\code v])	\\
			&\sqsubseteq_\Ie B_M([-]_\equiv)(g)	
		\end{align*}
		
		\item[(\acro{R3})] Suppose \(e = g\) and \(f = \mu\code v~e_1\), and assume that \(g \sqsubseteq \beta\code v~e_1[g\gdsub\code v]\) and
		\[
		B_M([-]_\equiv) \circ \epsilon(g) 
		\sqsubseteq_\Ie B_M([-]_\equiv) \circ \epsilon(\beta\code v~e_1[g\gdsub\code v])
		\]
		Then we have 
		\begin{align*}
			B_M([-]_\equiv) \circ \epsilon(g) 
			&\le B_M([-]_\equiv) \circ \epsilon(\beta\code v~e_1[g\gdsub\code v]) \\
			&= B_M([-]_\equiv)(\lfp_{\code v}\epsilon(e_1)[g\gdsub\code v]) \\
			&= B_M([-]_\equiv)(\lfp_{\code v}\epsilon(e_1))[[g]_\equiv\gdsub\code v] \\
			&\sqsubseteq_\Ie B_M([-]_\equiv)(\lfp_{\code v}\epsilon(e_1))[[\mu\code v~e_1]_\equiv\gdsub\code v] \\
			&= B_M([-]_\equiv)(\lfp_{\code v}\epsilon(e_1)[\mu\code v~e_1\gdsub\code v]) \\
			&= B_M([-]_\equiv) \circ \epsilon(\mu\code v~e_1)\qedhere
		\end{align*}
	\end{description}
\end{proof}

\paragraph*{Proof of Completeness}
We now turn our attention to the details of steps 1 through 4 in \cref{sec:axiomatisation}, beginning with the association of each ordered process with a system of equations.
Given a poset \((X, \le)\) of \emph{indeterminates}, the underlying set of the poset \(\Exp(X, \le)\) is given by the grammar
\[
x \mid \code v \mid \code a e \mid \sigma(e_1, \dots, e_n) \mid \beta \code v~ e \mid \mu\code v~e
\qquad\qquad (x \in X)
\]
and partially ordered by the smallest precongruence extending \(\le_{\exp}\) and the partial order \(\le\) on \(X\).
For any \(h : (X, \le) \to (\Exp, \le_{\exp})\), we define a map \(h^\# : \Exp(X, \le) \to (\Exp, \le_{\exp})\) by simultaneously substituting indeterminates for expressions, ie.
\[
	h^\#(x) = h(x) 
	\qquad 
	h^\#(\code v) = \code v 
	\qquad 
	h^\#(\code ae) = \code a~h^\#(e)
	\qquad 
	h^\#(\sigma(\vec e)) = \sigma(h^\#(e_1), \dots, h^\#(e_n))
\]
and \(h^\#(\beta\code v~e) = \beta\code v~h^\#(e)\) and \(h^\#(\mu\code v~e) = \mu\code v~h^\#(e)\) for any \(x \in X\) and \(e,e_i \in \Exp(X, \le)\). 

\begin{definition}
	A \defn{(monotone) system of equations (with indeterminates from)} \((X, \le)\) is a map \(s : (X, \le) \to \Exp(X, \le)\).
	A map \(\phi : (X, \le) \to (\E, \sqsubseteq)\) is a \emph{solution} to \(s\) if \(\phi \equiv \phi^{\#}\circ s\). 
	We say that \(s\) is \emph{guarded} if \(x\) is guarded in \(s(x)\) for all \(x \in X\). 
\end{definition}

Step 1 can now be completed as follows:
Given an ordered process \((X, \le, \vartheta)\) and a state \(x \in X\), we define \(s_\vartheta(x)\) by choosing a representative \(p(\vec{\code v}, \langle \code a_i, x_i\rangle_{i \le n})\) of \(\vartheta(x)\) and setting \(s_\vartheta(x) = p(\vec{\code v}, \code a_1x_1, \dots, \code a_nx_n)\). 
Any two choices in representatives produce equivalent syntactic systems up to \(\equiv\), so we simply refer to any \(s_\vartheta\) constructed in this manner as \emph{the} syntactic system associated with \((X, \le, \vartheta)\).
Notice also that \(s_\vartheta\) is necessarily guarded. 

Step 2 is contained in the following theorem, which implies that the behaviours of finite ordered processes can be compared (in \(\E\)) by solving their systems and comparing terms. 

\begin{theorem}\label{thm:solutions are homs}
	Let \((X, \le, \vartheta)\) be an ordered process with associated system of equations \(s_\vartheta\). 
	Then \(\phi : (X, \le) \to (\E, \sqsubseteq)\) is a solution to \(s_\vartheta\) if and only if \(\phi : {(X, \le, \vartheta)} \to (\E, \sqsubseteq, \bar\epsilon)\).
\end{theorem}

The proof of this is made much easier by the following lemma.

\begin{lemma}[Fundamental]
	The coalgebraic structure map \(\bar\epsilon : (\E, \sqsubseteq) \to B_M(\E,\sqsubseteq)\) has a monotone inverse.
\end{lemma}

\begin{proof}
	By (\acro{\Ie}), \((\E, \sqsubseteq, \bar\ev|_S) \models \Ie\). 
	Thus, the monotone map \((-)^\heartsuit : (\Var + \Act \times \E, \le) \to (\E, \sqsubseteq)\) defined by \[
		\code v^\heartsuit = \code v 
		\qquad \qquad
		\code \langle\code a, [e]_\equiv\rangle^\heartsuit = [\code ae]_\equiv
	\]
	lifts to a unique \(S\)-algebra homomorphism \((-)^\heartsuit : (B_M(\E, \sqsubset), \rho )\to (\E, \sqsubseteq, \bar\ev|_S)\).
	This map is readily seen to be a right inverse to \(\bar\epsilon\), as 
	\(
		\bar\epsilon(\code v^\heartsuit) = \bar\epsilon(\code v) = \code v
	\)
	and 
	\(
		\bar\epsilon(\langle\code a, [e]_\equiv\rangle^\heartsuit) 
		= \bar\epsilon([\code a e]_\equiv)
		= \langle \code a, [e]_\equiv\rangle 
	\)
	and both \(\bar\epsilon\) and \((-)^\heartsuit\) are \(S\)-algebra homomorphisms.
	It remains to see that \((-)^\heartsuit\) is a left inverse to \(\bar\epsilon\).
	We prove that \(\bar\epsilon([e]_\equiv)^\heartsuit = [e]_\equiv\) by induction on \(e\).
	\begin{itemize}
		\item \(\bar\epsilon([\code v]_\equiv)^\heartsuit = \code v^\heartsuit = \code v\)
		\item \(\bar\epsilon([\code ae]_\equiv)^\heartsuit = \langle \code a, [e]_\equiv\rangle^\heartsuit = \code a[e]_\equiv = [\code ae]_\equiv\)
	\end{itemize}
	This concludes the base cases. 
	For the inductive step,
	\begin{itemize}
		\item \(
		\bar\epsilon([\sigma(e_1, \dots, e_n)]_\equiv)^\heartsuit
		= \sigma(\bar\epsilon([e_1]_\equiv)^\heartsuit, \dots, \bar\epsilon([e_n]_\equiv)^\heartsuit)
		= \sigma([e_1]_\equiv, \dots, [e_n]_\equiv)
		= [\sigma(e_1, \dots, e_n)]_\equiv
		\)
		
		\item For the unguarded recursion case, let \(\bar\epsilon(e) = q(\code v, \vec{\code u}, \langle \code a_i, f_i\rangle_{i \le n})\).
		By induction, 
		\[
		[e]_\equiv = \bar\epsilon([e]_\equiv)^\heartsuit 
		= (q(\code v, \vec{\code u}, \langle \code a_i, [f_i]_\equiv\rangle_{i \le n}))^\heartsuit
		= q(\code v, \vec{\code u}, (\code a_i [f_i]_\equiv)_{i \le n})
		\]
		Thus, we have
		\begin{align*}
			\bar\epsilon([\beta\code v~e]_\equiv)^\heartsuit
			&= (\lfp_{\code v} q(\vec{\code u}, \langle \code a_i, [f_i]_\equiv\rangle_{i \le n}))^\heartsuit \\
			&= \lfp_{\code v} q(\vec{\code u}, \langle \code a_i, [f_i]_\equiv\rangle_{i \le n}^\heartsuit) \\
			&= \lfp_{\code v} q(\vec{\code u}, ([\code a_i f_i]_\equiv)_{i \le n}) \\
			&= \beta{\code v}~ q(\code v, \vec{\code u}, ([\code a_i f_i]_\equiv)_{i \le n}) \\
			&= [\beta\code v~e]_\equiv
		\end{align*}
		as desired.
		
		\item For the full recursion case, we need to first show that \((-)^\heartsuit\) commutes with \([[g]_\equiv\gdsub\code v]\) for any \(g \in \Exp\). 
		This is achieved by recalling that \([[g]_\equiv\gdsub\code v]\) is an \(S\)-algebra homomorphism \(B_M(\E, \subseteq) \to B_M(\E, \sqsubseteq)\) and noting that \[
			([\code v]_\equiv[[g]_\equiv\gdsub\code v])^\heartsuit = ([\code v]_\equiv)^\heartsuit = [\code v]_\equiv = [\code v]_\equiv^\heartsuit[[g]_\equiv\gdsub\code v]
		\]
		and
		\begin{align*}
		(\langle \code a, [e]_\equiv\rangle [[g]_\equiv\gdsub\code v])^\heartsuit 
		&= (\langle \code a, [e[g\sub\code v]]_\equiv\rangle)^\heartsuit 	\\
		&= [\code a e[g\sub\code v]]_\equiv \\
		&= [\code a e]_\equiv [[g]_\equiv\gdsub\code v] \\
		&= \langle\code a, [e]_\equiv\rangle^\heartsuit [[g]_\equiv\gdsub\code v]
		\end{align*}
		Finally, we use the previous case to conclude that
		\begin{align*}
			\bar\epsilon([\mu\code v~e]_\equiv)^\heartsuit 
			&= (\bar\epsilon([\beta\code v~e]_\equiv)[[\mu\code v~e]_\equiv\gdsub\code v])^\heartsuit \\
			&= (\bar\epsilon([\beta\code v~e]_\equiv))^\heartsuit [[\mu\code v~e]_\equiv\gdsub\code v] \\
			&= [\beta\code v~e]_\equiv[[\mu\code v~e]_\equiv\gdsub\code v] \\
			&= [\beta\code v~e[\mu\code v~e\gdsub\code v]]_\equiv \\
			&= [\mu\code v~e]_\equiv\qedhere
		\end{align*}
	\end{itemize}
\end{proof}

\begin{proof}[Proof of \cref{thm:solutions are homs}]
	It suffices to show that \(\phi\) is a solution iff \((-)^\heartsuit \circ B_M(\phi) \circ \vartheta = \phi\), since \(\bar\epsilon\) is a monotone left inverse to \((-)^\heartsuit\).
	This is straight-forward: for any \(x \in X\), observe that
	\begin{align*}
		(-)^\heartsuit\circ B_M(\phi)\circ\vartheta(x)
		&= (B_M(\phi) \circ \vartheta(x))^\heartsuit \\
		&= (\vartheta(x)[\phi(z)\gdsub z]_{z \in X})^\heartsuit \\
		&= (\vartheta(x))^\heartsuit[\phi(z)\gdsub z]_{z \in X} \\
		&= s_\vartheta(x)[\phi(z)\gdsub z]_{z \in X}
	\end{align*}
	The last expression is equal to \(\phi(x)\) if \(\phi\) is a solution, and the first expression is equal to \(\phi(x)\) if \(\phi\) is a homomorphism. 
\end{proof}

Step 3 takes the form of the following theorem, a variation on a result of Milner~\cite[Theorem 5.7]{milner1984complete} that establishes the unique solvability of systems of equations associated with finite LTSs. 
Extra care is needed in our result to ensure the monotonicty of our solutions.

\begin{theorem}\label{thm:ordered Milner theorem}
	Every finite guarded system of equations admits a unique solution.
\end{theorem}

The proof of this theorem requires the following lemma (particularly the last statement), which establishes the crutial features of the inference system given in \cref{fig:axioms}.

\begin{lemma}
	The following are derivable from the axioms in \cref{fig:axioms}.
	\begin{description}
		\item[\acro{(R3')}] If \(e[g\sub \code v] \sqsubseteq g\), then \(\mu\code v~e \sqsubseteq g\).
		\item[\acro{(R\alpha a)}] If \(\code w \notin \fv(f_i)\) for any \(i \le n\), then \(\beta\code w~p(\code w, \vec f) \sqsubseteq \beta\code v~p(\code v, \vec f)\).
		\item[\acro{(R\alpha b)}] If \(\code w \notin \fv(e)\), then \(\mu\code w~e[\code w/\code v] \sqsubseteq \mu\code v~e\).
		\item[\acro{(Rfp1)}] \(\beta \code v~e \equiv e[\beta\code v~e\unsub\code v]\) where \([\beta\code v~e\unsub\code v]\) replaces unguarded instances of \(\code v\) with \(\beta\code v~e\)
		\item[\acro{(Rfp2)}] \(\mu \code v~e \equiv e[\mu\code v~e\sub\code v]\) 
		\item[\acro{(Rex)}] \(\mu\code v~\beta\code v~e \equiv \mu\code v~e\)
		\item[\acro{(UFP)}] If \(\code v \in \gd(e)\) and \(e[g\sub\code v] \equiv g\), then \(\mu\code v~e \equiv g\). 
	\end{description}
\end{lemma}

\begin{proof}
	We show (\textsf{UFP}) because it is made use of in the next theorem.

	Since \(\code v \in \gd(e)\), \(e \equiv \beta\code v~e\), so \(e[g\gdsub\code v] = e[g \sub \code v] \equiv (\beta\code v~e)[g\sub\code v]\).
	From (\textsf{R3}) we see that \(g \sqsubseteq \mu\code v~e\).
	The converse inequality can be seen from (\textsf{R2b}). 
\end{proof}

We use the fact (\acro{UFP}) above frequently in the proof of \cref{thm:ordered Milner theorem}. 

\begin{proof}[Proof of \cref{thm:ordered Milner theorem}]
	Let \(X = \{x_1, \dots, x_n\}\) and set \(e_i = s(x_i) \in \Exp(X, \le)\) for each \(i \le n\). 
	We proceed by induction on \(n\), the base case being easy: 
	if \(n = 1\), then \(\phi(x_1) = \mu x_1~e_1\) is the unique solution by (\acro{UFP}).
	
	Now suppose that any guarded system with fewer than \(n\) indeterminates has a unique solution.
	Since \(X\) is finite, there is a minimal element of \((X, \le)\). 
	We take this to be the last of the indexed indeterminates, \(x_n\), and solve the system \(s\) by first solving for \(x_n\).
	
	To this end, let \(f_n = \mu x_n ~e_n\), and \(f_i = e_i[f_j\sub x_j]_{j \le n}\) for \(i < n\).
	Then by monotonicity of substitution, \(s' : ((X, \le)-x_n) \to \Exp((X,\le)-x_n)\) defined \(s'(x_i) = f_i\) for \(i < n\) defines a monotone syntactic system. 
	And furthermore, \(x_i\) is guarded in \(s'(x_i)\) for each \(i < n\). 
	By the induction hypothesis, there exists a unique solution \(\phi' : ((X, \le)-x_n) \to (\E, \sqsubseteq)\). 
	Let \(g_i = \phi'(x_i)\) for \(i < n\) and define \(g_n = f_n[g_i\sub x_i]_{i < n}\). 
	By assumption, \(\phi'\) is monotone. 
	If \(x_n \le x_i\), we know that \(e_n \sqsubseteq e_i\), and so by monotonicity of substitution
	\begin{align*}
		g_n &= f_n[g_1\sub x_1, \dots, g_{n-1}\sub x_{n-1}] \\
			&= \mu x_n~e_n[g_1\sub x_1, \dots, g_{n-1}\sub x_{n-1}] \\
			&\equiv e_n[f_n\sub x_n][g_1\sub x_1, \dots, g_{n-1}\sub x_{n-1}] \\
			&\equiv e_n[g_1\sub x_1, \dots, g_{n-1}\sub x_{n-1}, g_n\sub x_n] \\
			&\sqsubseteq e_i[g_1\sub x_1, \dots, g_{n-1}\sub x_{n-1}, g_n\sub x_n] \\
			&\equiv g_i
	\end{align*}
	This establishes monotonicity of \(\phi : (X, \le) \to (\E, \sqsubseteq)\), where \(\phi(x_i) = g_i\) for \(i \le n\).
	We have \(\phi(x_i) \equiv e_i[g_j\sub x_j]_{j \le n}\) for each \(i \le n\), and therefore \(\phi : (X, \le)\to (\E, \sqsubseteq)\) is a solution to \(s\). 
	
	Now let \(\psi :(X, \le) \to (\E, \sqsubseteq)\) be another solution to \(s\), and write \(h_i = \psi(x_i)\) for \(i \le n\).
	Then in particular, 
	\[
		h_n \equiv e_n[h_1\sub x_1, \dots, h_{n-1}\sub x_{n-1}, h_n\sub x_n]
	\]
	so by (\acro{UFP}), 
	\[
	h_n \equiv \mu x_n~ e_n[h_1\sub x_1, \dots, h_{n-1}\sub x_{n-1}]
	= f_n[h_1\sub x_1, \dots, h_{n-1}\sub x_{n-1}]
	\]
	It follows that \(\psi' : ((X, \le)-x_n) \to (\E, \sqsubseteq)\) defined \(\psi'(x_i) = h_i[f_n\sub x_n]\) is a solution to \(s'\). 
	By the induction hypothesis, \(\phi' \equiv \psi'\), or in other words \(g_i \equiv h_i\) for each \(i < n\). 
	By \(\Sigma_S\)-precongruence, 
	\begin{align*}
		h_n 
		&\equiv \mu x_n~e_n[h_1\sub x_1, \dots, h_{n-1} \sub x_{n-1}] \\
		&\equiv \mu x_n~e_n[g_1\sub x_1, \dots, g_{n-1} \sub x_{n-1}] \\
		&= g_n
	\end{align*}
	Thus, \(\phi \equiv \psi\), and consequently \(\phi\) is the unique solution.
\end{proof}

Finally, step 4 is a result of the following observation:
Call an ordered process \((X, \le, \vartheta)\) \emph{locally finite} if for every \(x \in X\) there is a finite subcoalgebra \(U \subseteq X\) of \((X \le, \vartheta)\) with \(x \in U\).

\begin{restatable}{lemma}{localfiniteness}\label{lem:local finiteness}
	Both \((\Exp, \le_{\exp}, \ev)\) and \((J, \le_\beh, \zeta_J)\) are locally finite. 
\end{restatable}

To prove \cref{lem:local finiteness}, we need the following.

\begin{lemma}\label{lem:lfp reduces variables}
	Let \(p(x, \vec y) \in M(X, \le)\) and \(q = \lfp_xp\).
	Then \(q = q(\vec y)\).
\end{lemma}

\begin{proof}
	Let \(z \notin \vec y\).
	It suffices to show that if \(q = q(\vec y, z)\) for some \(z \in X\setminus\{x\}\), then for any \(r \in M(X, \le)\), \(q(\vec y, r) \equiv_\Ie q\).
	Since \(p(q, \vec y) \sqsubseteq_\Ie q\), we know that \(p(q(\vec y, r), \vec y) \sqsubseteq_\Ie q(\vec y, r)\),
	so that \(q(\vec y, r) \in \Pfp_x(p)\). 
	By assumption, \(q \sqsubseteq_\Ie q(\vec y, r)\). 
	But \(r\) is arbitrary, so in particular \(q(\vec y, z) \sqsubseteq_\Ie q(\vec y, 0)\). 
	Replacing \(z\) with an arbitrary \(r\) again reveals \(q(\vec y, r) \sqsubseteq_\Ie q(\vec y, 0)\), so in fact, \(q(\vec y, r) \equiv_\Ie q(\vec y, 0)\), and furthermore \(q(\vec y, r_1) \equiv_\Ie q(\vec y, r_2)\) for any \(r_1, r_2 \in M(X, \le)\).
	It follows that \(q(\vec y, z) \equiv_\Ie q(\vec y, r)\) for all \(r \in M(X, \le)\). 
\end{proof}

\begin{proof}[Proof of \cref{lem:local finiteness}]
	It suffices to see that \((\Exp, \le_{\exp}, \epsilon)\) is locally finite. 
	Given an arbitrary \(e \in \Exp\), we will explicitly construct a subcoalgebra of \((\Exp, \le_{\exp}, \epsilon)\) that has a finite set of states that includes \(e\).
	To this end, we define a function \(U : \Exp \to \P(\Exp)\) inductively.
	For any \(e \in \Exp\), write \(e[0\unsub \code v]\) for the expression obtained by replacing every unguarded instance of \(\code v\) in \(e\) with \(0\).
	Then   
	\begin{gather*}
		U(\code v) = \{\code v\}
		\qquad 
		U(\code a e) = \{\code a e\} \cup U(e)
		\qquad 
		U(\sigma(e_1, \dots, e_n)) = \{\sigma(e_1, \dots, e_n)\} \cup \bigcup_{i \le n} U(e_i)
		\\
		U(\beta\code v~e) = \{\beta\code v~e\} \cup U(e[0\unsub \code v])
		\qquad 
		U(\mu\code v~e) = \{\mu\code v~e\}\cup \{f[\mu\code v~e \gdsub \code v] \mid f \in U(e[0\unsub \code v])\}
	\end{gather*}
	Note that \(e \in U(e)\) for all \(e \in \Exp\), and that \(U(e)\) is finite. 
	We begin with following claim, which says that the derivatives of \(e\) can be given in terms of expressions from \(U(e)\): 
	for any \(e \in \Exp\), there is a representative \(S\)-term \(p \in \epsilon(e)\) such that \(p \in S^*(\Var + \Act \times U(e))\). 
	This can be seen by induction on \(e\).
	For the base cases, we see that
	\begin{itemize}
		\item \(\epsilon(\code u) = \code u \in S^*(\Var + \Act \times U(\code u)) = S^*(\Var + \Act \times \{\code u\})\) if \(\code{u} \in \Var\)
		\item \(\epsilon(\code a e) = \langle \code a, e\rangle \in S^*(\Var + \Act\times U(\code a e))\) since \(e \in U(\code a e)\).
	\end{itemize}
	For the inductive cases, 
	\begin{itemize}
		\item \(\epsilon(\sigma(e_1, \dots, e_n)) = \sigma(\epsilon(e_1), \dots, \epsilon(e_n))\) and since \(\epsilon(e_i)\) has a representative \(p_i \in S^*(\Var + \Act\times U(e_i))\) for each \(i \le n\), \(\sigma(p_1, \dots, p_n) \in S^*(\Var + \Act\times \bigcup_{i\le n} U(e_i))\).
		\item \(\epsilon(\beta\code u~e) = \lfp_{\code u}\epsilon(e)\), so we need to show that if \(p(\vec{\code v}, \langle\code a_i, e_i\rangle_{i \le n}) \in S^*(\Var + \Act\times U(e))\) is a representative of \(\epsilon(e)\), then at most the variables \(\vec {\code v}\) and \(\langle\code a_i, e_i\rangle_{i \le n}\) need appear in a representative of \(\lfp_{\code u}\epsilon(e)\).
		This is the content of \cref{lem:lfp reduces variables}.
		\item \(\epsilon(\mu \code v~e) = \lfp_{\code v}\epsilon(e)[\mu\code v\gdsub\code v]\), in which case we let \(p \in \epsilon(e)\) and note that \(\lfp_{\code v} p[\mu \code v~e/\!/v]\) is a representative of \(\epsilon(\mu \code v~e)\) in \(S^*(\Var + \Act \times U(\mu v~e))\).
	\end{itemize}
	
	To finish the proof of the lemma, fix an \(e \in \Exp\) and define a sequence of sets beginning with \(U_0 = \{e\}\) and proceeding with 
	\[
	U_{n+1} = U_n \cup \bigcup_{e_0 \in U_n}\{g \mid (\exists \code a \in \Act)(\exists p \in \epsilon(e_0)\cap S^*(\Var + \Act\times U(e_0)))~\text{\(\langle\code a,g\rangle\) appears in \(p\)}\}
	\] 
	Then \(U_0 \subseteq U_1 \subseteq \cdots \subseteq U(e)\), and the latter set is finite. 
	Hence \(U := \bigcup U_n\) is finite and contained in \(U(e)\). 
	We define a coalgebra structure \(\epsilon_U : U \to B_MU\) by taking \(\epsilon_U(e) = p\) where if \(e \in U_n\), then \(p\) is a representative of \(\epsilon(e)\) in \(S^*(\Var + \Act\times U_{n+1})\).
	Since \(S^*(\Var + \Act \times U_{n+1}) \subseteq S^*(\Var + \Act \times U)\) and \(\epsilon_U\) agrees with \(\epsilon\) (up to \(\equiv_\Ie\)), this defines a monotone \(B_M\)-coalgebra structure on \(U\).
	Where \(\iota : U \hookrightarrow \Exp\), we have \(\epsilon(\iota(e)) = \epsilon(e) = B_M(\iota)\circ\epsilon_U(e)\). 
	Thus, \((U,\epsilon_U)\) is a finite subcoalgebra of \((\Exp,\le_{exp},\epsilon)\) containing \(e\). 
\end{proof}

\begin{proof}[Proof of \cref{thm:completeness}]
	The formal construction of \(\phi : J \to \E\) can now be carried out by choosing a finite subcoalgebra \(U_t \subseteq J\) containing \(t\) for each \(t \in J\) and setting \(\phi(t) = \phi_{U_t}(t)\) where \(\phi_{U_t}\) is the unique solution to the system of equations associated with \((U_t, \le_\beh, \zeta_{U_t})\). 
	Every homomorphism out of a locally finite coalgebra is the union of its restrictions to finite subcoalgebras, so it suffices to check that \(\phi\) is well-defined and restricts to a coalgebra homomorphism on each finite subcoalgebra.

	To that end, notice that if \(t \in U_{t_1} \cap U_{t_2}\) we can write \(\zeta(t) = p(\vec {\code v}, \langle \code a_i, s_i\rangle_{i \le n})\) with \(s_i \in U_{t_1} \cap U_{t_2}\), and similarly for each \(s_i\) and so on, so that there is a finite subcoalgebra \(U \subseteq U_{t_1} \cap U_{t_2}\) containing \(t\). 
	By uniqueness of solutions, \(\phi_{U_{t_1}}(t) \equiv \phi_{U}(t) \equiv \phi_{U_{t_2}}(t)\), so \(\phi\) is a well-defined coalgebra homomorphism. 
	To see that \(\phi\) is left-inverse to \(\sem-\), observe that
	\(
		\phi(\sem e) = \phi_{\sem e}(\sem e) \equiv e
	\) 
	follows from \cref{thm:ordered Milner theorem} since the inclusion map \(\phi[U_{\sem e}] \hookrightarrow \E\) is the unique solution to the finite subcoalgebra \(\phi[U_{\sem e}]\) of \((\E, \sqsubseteq, \bar\epsilon)\). 
\end{proof}

\paragraph*{An aside: Least solutions without guardedness}
There is another version of \cref{thm:ordered Milner theorem} that does not require guardedness of systems.
The following proposition records this, but it should be noted that we do not require this proposition anywhere else in the appendix. 

\begin{proposition}
	Every finite system of equations admits a least solution. 
\end{proposition}

\begin{proof}
	The idea behind this proof is the same as the proof of \cref{thm:ordered Milner theorem}, but with a different induction hypothesis/step.
	Let \(X = \{x_1, \dots, x_n\}\) and set \(e_i = s(x_i) \in \Exp(X, \le)\) for each \(i \le n\). 
	We will show, inductively on \(n\), that (i) \(s\) has a solution \(\phi\), and (ii) for every \emph{pre}solution \(\psi : {(X, \le)} \to {(\E, \sqsubseteq)}\) of \(s\), meaning that \(\psi^\# \circ s \sqsubseteq \psi\), we have \(\psi(x_i) \sqsubseteq \psi(x_i)\) for each \(i \le n\). 

	Again, the base case is easy: 
	if \(n = 1\), then \(\phi(x_1) = \mu x_1~e_1\) is a solution.
	Furthermore, if \(e_1[\psi(x_1) \sub x_1] \sqsubseteq \psi(x_1)\), then \(\phi(x_1) \sqsubseteq \psi(x_1)\) by the least fixed-point axiom of \cref{fig:axioms}.
	
	Now suppose that any guarded system with fewer than \(n\) indeterminates has a solution below every presolution.
	Since \(X\) is finite, there is a minimal element of \((X, \le)\). 
	Again, we take this to be the last of the indexed indeterminates, \(x_n\), and solve the system \(s\) by first solving for \(x_n\).
	
	To this end, let \(f_n = \mu x_n ~e_n\), and \(f_i = e_i[f_j\sub x_j]_{j \le n}\) for \(i < n\).
	Then by monotonicity of substitution, \(s' : ((X, \le)-x_n) \to \Exp((X,\le)-x_n)\) defined \(s'(x_i) = f_i\) for \(i < n\) defines a monotone syntactic system. 
	By the induction hypothesis, there exists a solution \(\phi' : ((X, \le)-x_n) \to (\E, \sqsubseteq)\) that is below every presolution. 
	Let \(g_i = \phi'(x_i)\) for \(i < n\) and define \(g_n = f_n[g_i\sub x_i]_{i < n}\). 
	By assumption, \(\phi'\) is monotone. 
	If \(x_n \le x_i\), we know that \(e_n \sqsubseteq e_i\), and so by monotonicity of substitution
	\begin{align*}
		g_n &= f_n[g_1\sub x_1, \dots, g_{n-1}\sub x_{n-1}] \\
			&= \mu x_n~e_n[g_1\sub x_1, \dots, g_{n-1}\sub x_{n-1}] \\
			&\equiv e_n[f_n\sub x_n][g_1\sub x_1, \dots, g_{n-1}\sub x_{n-1}] \\
			&\equiv e_n[g_1\sub x_1, \dots, g_{n-1}\sub x_{n-1}, g_n\sub x_n] \\
			&\sqsubseteq e_i[g_1\sub x_1, \dots, g_{n-1}\sub x_{n-1}, g_n\sub x_n] \\
			&\equiv g_i
	\end{align*}
	This establishes monotonicity of \(\phi : (X, \le) \to (\E, \sqsubseteq)\), where \(\phi(x_i) = g_i\) for \(i \le n\).
	We have \(\phi(x_i) \equiv e_i[g_j\sub x_j]_{j \le n}\) for each \(i \le n\), and therefore \(\phi : (X, \le)\to (\E, \sqsubseteq)\) is a solution to \(s\). 
	
	Now let \(\psi :(X, \le) \to (\E, \sqsubseteq)\) be any presolution to \(s\), and write \(h_i = \psi(x_i)\) for \(i \le n\).
	Then in particular, 
	\[
		h_n \sqsupseteq e_n[h_1\sub x_1, \dots, h_{n-1}\sub x_{n-1}, h_n\sub x_n]
	\]
	By the least fixed-point axioms of \cref{fig:axioms}, 
	\[
	h_n \sqsupseteq \mu x_n~ e_n[h_1\sub x_1, \dots, h_{n-1}\sub x_{n-1}]
	= f_n[h_1\sub x_1, \dots, h_{n-1}\sub x_{n-1}]
	\]
	It follows that \(\psi' : ((X, \le)-x_n) \to (\E, \sqsubseteq)\) defined \(\psi'(x_i) = h_i[f_n\sub x_n]\) is a presolution to \(s'\). 
	By the induction hypothesis, \(\phi' \sqsubseteq \psi'\), or in other words \(g_i \sqsubseteq h_i\) for each \(i < n\). 
	By \(\Sigma_S\)-precongruence, 
	\begin{align*}
		h_n 
		&\sqsupseteq \mu x_n~e_n[h_1\sub x_1, \dots, h_{n-1} \sub x_{n-1}] \\
		&\sqsupseteq \mu x_n~e_n[g_1\sub x_1, \dots, g_{n-1} \sub x_{n-1}] \\
		&= g_n
	\end{align*}
	Thus, \(\phi \sqsubseteq \psi\), and consequently \(\phi\) is a solution below every presolution.
\end{proof}

\section{The small-step semantics of polystar expressions}

Polystar expressions also carry an intuitive small-step semantics, given by an \(L_M\)-coalgebra structure \((\PExp, \le_{\text{pxp}}, \ell)\) where 
\(
L_M = M(\{\checkmark\} + \Act \times \Id)
\).
This is outlined below.

\begin{mathpar}
	\ell(0) = 0		
	\and 
	\ell(1) = 1			
	\and 
	\ell(e +_\sigma f) = \sigma(\ell(e), \ell(f))		
	\and 
	\ell(e f) = q(\ell(f), \langle \code a_1, e_1 f\rangle, \dots, \langle \code a_n, e_nf\rangle)		
	\and 
	\ell(\code a) = \langle \code a, 1\rangle
	\and
	\ell(e^{[p]}) = \lfp_x (p(q(x, \langle \code a_1, e_1e^{[p]}\rangle, \dots, \langle \code a_n, e_ne^{[p]}\rangle), \checkmark))
\end{mathpar}
Here, \(\sigma \in I\); \(e, e_1, e_2 \in \PExp\); \(\code a, \code a_i \in \Act\); \(p = p(x,y) \in S^*\{x,y\}\); \(\ell(e) = q(\checkmark, \langle \code a_1, e_1\rangle, \dots, \langle\code a_n, e_n\rangle)\). 

\section{Proofs from \cref{sec:simulations}}

\ordinaryimpliesordered*

\begin{proof}
	Let \(h : (X, \vartheta) \to (Y, \delta)\) be a \(B_N\)-coalgebra homomorphism such that \(h(x) = h(y)\).
	Then \(h = D(h) : (X, =, \vartheta) \to (Y, =, \delta)\), so \(x =_\beh y\).
\end{proof}

\orderedimpliesordinary*

\begin{proof}
	Suppose \(M\) lifts \(N\) and let \((X, \vartheta)\) be a \(B_N\)-coalgebra with states \(x, y \in X\). 
	If \(x \bisim y\), then \(x =_\beh y\) by \cref{lem:ordinary implies ordered}. 
	If \(x =_\beh y\), let \(h : (X, =, \vartheta) \to (Y, \le, \delta)\) be such that \(h(x) = h(y)\). 
	Then \(h\) is an \(N\)-coalgebra homomorphism, because \(U(\delta) : X \to UB_M(X, =) = B_N X\) and 
	\[
		\delta \circ h 
		= U(\delta \circ h) 
		= U(B_M(h) \circ \vartheta)
		= UB_M(h) \circ U(\vartheta)
		= B_ND(h) \circ \vartheta
		= B_N(h) \circ \vartheta
	\]
	It follows that \(x \bisim y\). 

	Conversely, suppose that \(x =_\beh y\) but not \(x \bisim y\). 
	There is an \(h : (X, =, \vartheta) \to (Y, \le, \delta)\) such that \(h(x) = h(y)\), but it must be the case that \(U(h) = h\) fails to be a \(B_N\)-coalgebra homomorphism.
	However, \(UB_M(h) \circ \vartheta = U(B_M(h) \circ \vartheta) = U(\delta \circ h) = U(\delta) \circ h\), so it must be the case that \(UB_M(h) \neq B_N(h) = B_NU(h)\).
	It follows that \(M\) fails to lift \(N\). 
\end{proof}

Say that \((S, \Ie)\) \emph{admits couplings} if for any \(h : (X, \le) \to (Y, \le)\) and any \(p(\vec x),q(\vec y) \in M(X, \le)\) such that \(p(h(\vec x)) \equiv_\Ie q(h(\vec y))\), there is an \(h\)-coupling \(r\) of \(p\) and \(q\). 
The next lemma is in reference to \cref{rem:admits couplings wpb}. 

\begin{lemma}
	\(M\) weakly preserves pullbacks if and only if \((S, \Ie)\) admits couplings.
\end{lemma}

\begin{proof}
	Obtain \(k_1 : (P, \le) \to (X, \le)\) and \(k_2 : (P, \le) \to (Y, \le)\) as the pullback of \(f_1 : (X, \le) \to (W, \le)\) and \(f_2 : (Y, \le) \to (W, \le)\).
	Also obtain \(l_1 : (Q, \le) \to M(X, \le)\) and \(l_2 : (Q, \le) \to M(Y, \le)\) as the pullback of \(M(f_1), M(f_2)\). 
	The statement that \(M\) weakly preserves pullbacks is equivalent to the statement that \(M(P, \le)\) always surjects onto \((Q, \le)\). 

	Now, \(Q \cong \{(p,q) \in M(X, \le) \times M(X, \le) \mid M(f_1)(p) = M(f_2)(q)\}\) and \(P \cong \{(x, y) \in X \mid f_1(x) = f_2(y)\}\). 
	In other words, if \(p = p(\vec x)\) and \(q = q(\vec y)\), \(p(f_1(\vec x)) \equiv_\Ie q(f_2(\vec y))\). 
	It follows that the function \(M(P, \le) \to (Q, \le)\) defined \(r((u_i,v_i)_{i \le n}) \mapsto (r(\vec u), r(\vec v))\) is monotone. 

	By \cite[Theorem 8.6]{gumm1999elements}, \(M\) preserves weak pullbacks if and only if there is a \(r((u_i,v_i)_{i \le n}) \in M(P, \le)\) such that \(r(\vec u) \equiv_\Ie p\) and \(r(\vec v) \equiv_\Ie q\) for every \((p,q) \in Q\).
	That is, if and only if \((S, \Ie)\) admits couplings.
\end{proof}

\behaviouralinequivalenceisasim*

\begin{proof}
    We aim to show that for any pair \((x,y) \in X^2\) satisfying \(x \le_\beh y\), there is an \[
		r(\vec {\code v}, \langle \code a_i, (x_i, y_i)\rangle_{i \le n}) \in B_M(\le_\beh)
	\]
	such that \(\theta(x) \sqsubseteq_\Ie r(\vec {\code v}, \langle \code a_i, x_i\rangle_{i \le n})\) and \(r(\vec {\code v}, \langle \code a_i, x_i\rangle_{i \le n})\sqsubseteq_\Ie \theta(y)\).
	To this end, let \(h : (X, =, \vartheta) \to (Y, \le, \delta)\) such that \(h(x) \le h(y)\). 
	Then \[
		B_M(h) \circ \vartheta(x)
		\equiv_\Ie \delta \circ h(x) 
		\sqsubseteq_\Ie \delta \circ h(y)
		\equiv_\Ie B_M(h) \circ \vartheta(y)
	\]
	by monotonicity. 
	By assumption, there is a weak \((\Var + \Act\times h)\)-coupling \(r(\vec {\code v}, \langle\code a_i, (u_i,v_i)\rangle_{i \le n})\) of \(\vartheta(x)\) with \(\vartheta(y)\).
	That is, we have \(h(u_i) \le h(v_i)\) for each \(i \le n\), and \(\vartheta(x) \sqsubseteq_\Ie r(\vec {\code v}, \langle\code a_i, u_i\rangle_{i \le n})\), and \(r(\vec {\code v}, \langle\code a_i, v_i\rangle_{i \le n}) \sqsubseteq_\Ie \vartheta(y)\). 
	Since \(h(u_i) \le h(v_i)\), \(u_i \le_\beh v_i\) for each \(i \le n\) and \(r \in B_M(\le_\beh)\) as desired. 
\end{proof}

\collapsetheorem*

\begin{proof}
	Let \(x\) and \(y\) be states of \((X, \vartheta)\). 
	It suffices to show that the determinacy condition ensures that \(x \sim y\) implies \(x =_\beh y\). 
	To this end, let \(R_1,R_2 \subseteq X^2\) be simulations witnessing \(x \precsim y\) and \(y \precsim x\) respectively. 
	Let \(\le\) be the reflexive-transitive closure of \(R_1 \cup R_2\) and \(Y = X/{\le}\). 
	Then the quotient map \(h : (X, =) \to (Y, \le)\) satisfies \(h(x) \le h(y)\) and \(h(y) \le h(x)\).
	Since \((S, \Ie)\) admits weak \(h\)-couplings in both directions, by detemrinacy there is an \(h\)-coupling \(r(\vec v, \langle \code a, (u_i,v_i)\rangle_{i \le n})\) of \(\theta(x)\) and \(\theta(y)\). 

	Now let \(R = \{(u, v) \in X^2 \mid h(u) = h(v)\}\) and define \(\rho : (R, =) \to B_M (R, =)\) by choosing an \(h\)-coupling \(r_{u,v}(\vec v, \langle \code a_i, (u_i,v_i)\rangle_{i \le n})\) of \(\vartheta(u)\) with \(\vartheta(v)\) and setting 
	\[
		\rho(u, v) \equiv_\Ie r_{u,v}(\vec v, \langle \code a_i, (u_i,v_i)\rangle_{i \le n})
	\]
	Now,
	\begin{align*}
		B_M(\pi_1) \circ \rho(u, v)
		&\equiv_{\Ie} B_M(\pi_1) (r_{u,v}(\vec v, \langle \code a_i, (u_i,v_i)\rangle_{i \le n})) \\
		&= r_{u,v}(\vec v, \langle \code a_i, u_i\rangle_{i \le n}) \\
		&\equiv_\Ie \vartheta(u) \\
		&= \vartheta \circ \pi_1(u,v)
	\end{align*}
	for any \((u,v) \in R\), and similarly \(B_M(\pi_2) \circ \rho(u, v) \equiv_\Ie \vartheta(v)\).
	This tells us that \((R, \rho)\) is a \emph{bisimulation} on \((X, \vartheta)\) containing \((x, y)\). 
	This implies that \(x \bisim y\)~\cite{rutten2000universal}, so by \cref{lem:ordinary implies ordered} we have \(x =_\beh y\) as desired. 
\end{proof}

\section{Ordered Convex Algebra}\label{app:ordered convex algebra}

In this section, we give a proof of the monad presentation claimed in \cref{eg:convex algebra}.

\begin{lemma}
	The heavier-higher order is a partial order.
\end{lemma}

\begin{proof}
	The heavier-higher order is clearly reflexive and transitive.
	To see that it is antisymmetric, let \(\theta_1 \le^{hh} \theta_2 \le^{hh} \theta_1\). 
	Then in particular, where \[
	\upper x := \{u \in X \mid x \le u\}
	\]
	we see that \(\theta_1(\upper x) = \theta_2(\upper x)\) for all \(x \in X\).	
	Since \(\supp(\theta_1)\) is finite, it has maximal elements. 
	Let \(x_0\) be such an element of \(X\), and observe that \(\theta_1(\upper x_0) = \theta_1(x_0)\).
	We proceed by induction on the number of elements of \(\supp(\theta_1)\).
	
	Suppose \(\supp(\theta_1) = \{x_0\}\).
	Since \(\theta_2(\upper x_0) = \theta_1(\upper x_0)\), \(\theta_2(\upper x_0) > 0\) and so \(\supp(\theta_2) \cap \upper x_0\) is nonempty.
	Let \(x_1\) be a maximal element of \(\supp(\theta_2) \cap \upper x_0\).
	Then \(x_0 \le x_1\) and \(\theta_2(\upper x_1) = \theta_2(x_1)\).
	Hence, \[
	\theta_2(x_1) = \theta_2(\upper x_1) = \theta_1(\upper x_1) = \theta_1(\upper x_0) = \theta_1(x_0)
	\] 
	Since this implies that \(x_1 \in \supp(\theta_1)\), \(x_0 = x_1\).
	It follows that \(\theta_1(x_0) = \theta_2(x_0)\).
	
	Now define \(\theta_i'(x) = \theta_i(x)\) when \(x \neq x_0\) and \(\theta_i(x) = 0\) otherwise, \(i = 1,2\).
	Since \(\theta_i = {\theta_i(x_0) \delta_{x_0}} + \theta_i'\) for \(i = 1,2\), by cancellativity of \(\mathbf P\) we know that 
	\[
		\theta_1(x_0) \delta_{x_0} + \theta_1' 
		\le^{hh} \theta_2(x_0) \delta_{x_0} + \theta_2'
		\le^{hh} \theta_1(x_0) \delta_{x_0} + \theta_1'
	\]
	implies \(\theta_1' \le^{hh} \theta_2' \le^{hh} \theta_1'\), because \(\theta_1(x_0) = \theta_2(x_0)\).
	It follows by induction that \(\theta_1' = \theta_2'\).
	Therefore \(\theta_1 = \theta_2\) as well.
\end{proof}

\begin{lemma}
	For any monotone map \(f : (X, \le) \to (Y, \le)\), the map \(f^\bullet : \D X \to \D Y\) is monotone wrt to the heavier-higher order. 
\end{lemma}

\begin{proof}
	Let \(U \subseteq Y\) be an upper set of \((Y, \le)\).
	Then \(f^{-1}(U)\) is an upper set of \((X, \le)\).
	We have
	\begin{align*}
		f^\bullet\theta_1(U) &= \sum_{u \in U} f^\bullet\theta_1(u) 
		= \sum_{u \in U} \sum_{f(x) = u} \theta_1(x) 
		= \sum_{f(x) \in U} \theta_1(x) 
		\le \sum_{f(x) \in U} \theta_2(x) 
		= f^\bullet\theta_2(U) 
	\end{align*}
\end{proof}

That \((\D, \delta, \oplus)\) is the free convex algebra follows directly from the following three facts. 

\begin{description}
	\item[Fact (1)] The transformation \(\delta\) is monotone with respect to the heavier-higher order on any given base partial order. 
	That is, for any \((X, \le)\), \(\delta : (X, {\le}) \to (\D X, {\le^{hh}})\).
	\item[Fact (2)] \((\D X, \le^{hh}, \oplus)\) is an ordered convex algebra.
	It suffices to see that given \(\theta_1,\theta_2, \varphi \in \D X\) and \(r, s \in [0,1]\),
	\begin{itemize}
		\item[(i)] if \(\theta_1 \le^{hh} \theta_2\), then \(\theta_1 + \varphi \le^{hh} \theta_2 +\varphi\), and
		\item[(ii)] if \(r \le s\), then \(p \cdot \varphi \le^{hh} q \cdot \varphi\).
	\end{itemize} 
	\item[Fact (3)] For any ordered convex algebra \((Y, \le, \beta)\) and any \(f : (X, \le) \to (Y, \le)\), the inductively specified convex algebra homomorphism \(f^\beta : (\D, \rho) \to (Y, \beta)\) is monotone with respect to the heavier-higher order, ie. \(f^\beta : (\D X, \le^{hh}, \rho) \to (Y, \le, \beta)\).
\end{description}

We are going to argue for Fact (3), as it has a nontrivial proof.



\begin{proof}[Proof (of Fact (3)).]
	Let \(\theta_1 \le^{hh} \theta_2\) in \(\D (X, \le)\), and \(f : (X, \le) \to (Y, \le)\) where \((Y, \le, \beta)\) is a convex algebra.
	We proceed by induction on \(|\supp(\theta_1)|\). 
	If \(\supp(\theta_1) = \emptyset\), then \(\theta_1 = 0 \le^{hh} \theta_2\). 
	
	For the induction step, let \(z\) be any maximal element of \(\supp(\theta_1)\) and define 
	\[
		\theta_i'(x) = \begin{cases}
			\theta_i(z) - \theta_1(z) & x = z \\
			\theta_i(x) &\text{o.w.}
		\end{cases}
		\tag{i = 1,2}
	\]
	For any upper set \(U\) containing \(z\), \(\theta_i(U) = \theta_1(z) + \theta_i'(U)\), so by assumption \(\theta_1(z) + \theta_1'(U) \le \theta_1(z) + \theta_2'(U)\).
	Since \((\mathbb R, +)\) is cancellative, \(\theta_1'(U) \le \theta_2'(U)\).
	For any upper set \(U\) not containing \(z\), \(\theta_i(U) = \theta_i'(U)\), so by again, \(\theta_1'(U) \le \theta_2'(U)\).
	It follows that \(\theta_1' \le^{hh} \theta_2'\), so by induction \(f^\beta(\theta_1') \le f^\beta(\theta_2')\).
	We therefore have \[
		f^\beta(\theta_1) 
		= f^\beta(\theta_1(z)\cdot \delta_z) + f^\beta(\theta_1') 
		\le f^\beta(\theta_1(z)\cdot \delta_z) + f^\beta(\theta_2')
		\le f^\beta(\theta_1(z)\cdot \delta_z) + f^\beta(\theta_2')
		= f^\beta(\theta_2)
	\]
	This shows that \(f^\beta\) is monotone. 
\end{proof}
\end{document}